\documentclass[envcountsame,envcountsect]{llncs}
\usepackage[english]{babel}
\usepackage{amsmath}
\usepackage{amsfonts}
\usepackage{amssymb}
\usepackage{hyperref}
\usepackage{tikz}
\usepackage{stmaryrd}
\usepackage{mathtools}
\usepackage{cite}
\usepackage{stix}
\usepackage[section]{placeins}
\usetikzlibrary{positioning,arrows.meta,decorations.markings,decorations.pathreplacing,bending,calc,decorations.pathmorphing,shapes}

\title{Event structures for the reversible early internal $\pi$-calculus}
\author{Eva Graversen \and Iain Phillips\orcidID{0000-0001-5013-5876} \and Nobuko Yoshida \orcidID{0000-0002-3925-8557}}
\institute{Imperial College London, UK}

\newcommand{\fn}{\mathsf{fn}}
\newcommand{\bn}{\mathsf{bn}}
\newcommand{\M}[1]{\mathcal{#1}}

\newcommand{\CS}[1]{(E_{#1},F_{#1},\mathsf{C}_{#1},\rightarrow_{#1})}
\newcommand{\CStrans}[2]{\xrightarrow{#1\cup \underline{#2}}}

\newcommand{\cf}{\mathrel{\sharp}}

\newcommand{\I}{\mathsf{Init}}
\newcommand{\ind}{\;\;\;}

\newcommand{\LRBES}[1]{(E_{#1},F_{#1},\mapsto_{#1},\cf _{#1},\rhd_{#1},\lambda_{#1},\mathsf{Act}_{#1})}

\newcommand{\lrangles}[1]{\left\langle {#1} \right\rangle}
\newcommand{\lrBraces}[1]{\left\lBrace {#1} \right\rBrace}

\newcommand{\rfrac}[2]{{}^{#1}\!/_{#2}}
\newcommand{\inp}[2]{#1(#2)}
\newcommand{\outp}[2]{\overline{#1}(#2)}

\newcommand{\loc}{\mathsf{loc}}
\usepackage{xfrac}
\usepackage{bussproofs}
\usepackage{proof}
\usepackage{multicol}
 \usepackage{vwcol}  
\begin{document}
\tikzset{->-/.style={decoration={
  markings,
  mark=at position #1 with {\arrow{>}}},postaction={decorate}}}
  
  \newcounter{sarrow}
  \newcommand\xrsquigarrow[1]{\hspace{1pt}%
\mathrel{\begin{tikzpicture}[baseline={($(current bounding box.south)+(0,-0.5ex)$)}]
\node[inner sep=.5ex] (\thesarrow) {$\scriptstyle #1$};
\draw[<-,decorate,decoration={snake,amplitude=0.7pt,segment length=1.2mm,pre=lineto,pre length=4pt}] (\thesarrow.south east) -- (\thesarrow.south west);
\end{tikzpicture}}%
\hspace{1pt}}
\newcommand\xsquigarrow[2]{\hspace{1pt}%
\stepcounter{sarrow}%
\mathrel{\begin{tikzpicture}[anchor=base,baseline=-8pt]
\node[inner sep=.5ex] (\thesarrow) {$\scriptstyle #1$};
\node[inner sep=.5ex,below=0.85ex of \thesarrow] (below) {$\scriptstyle #2$};
\path[draw,<-,decorate,
  decoration={zigzag,amplitude=0.7pt,segment length=1.2mm,pre=lineto,pre length=4pt},transform canvas={yshift = -1pt}] 
    (\thesarrow.south east) -- (\thesarrow.south west);
\end{tikzpicture}}%
\hspace{1pt}}
\tikzset{negated/.style={
        decoration={markings,
            mark= at position 0.5 with {
                \node[transform shape] (tempnode) {$\backslash$};
            }
        },
        postaction={decorate}
    }
}\newcommand\xarrowtail[2]{\hspace{1pt}%
\stepcounter{sarrow}%
\mathrel{\begin{tikzpicture}[anchor=base,baseline=-8pt]
\node[inner sep=.5ex] (\thesarrow) {$\scriptstyle #1$};
\node[inner sep=.5ex,below=0.8ex of \thesarrow] (below) {$\scriptstyle #2$};
\path[draw,<-<] 
    (\thesarrow.south east) -- (\thesarrow.south west);
\end{tikzpicture}}%
\hspace{1pt}}
\tikzset{negated/.style={
        decoration={markings,
            mark= at position 0.5 with {
                \node[transform shape] (tempnode) {$\backslash$};
            }
        },
        postaction={decorate}
    }
}
\maketitle
\begin{abstract}
%
%
The $\pi$-calculus is a widely used process calculus, which models communications between processes 
and allows the passing of communication links.
Various operational semantics of the $\pi$-calculus have been proposed, which can be classified according to whether transitions are unlabelled (so-called reductions) or labelled.
With labelled transitions, we can distinguish early and late semantics.
The early version allows a process to receive names it already knows from the environment, while the late semantics and reduction semantics do not.
All existing reversible versions of the $\pi$-calculus
use reduction or late semantics,
despite
the early semantics of the (forward-only) $\pi$-calculus being more widely used than the late.
We define $\pi$IH, the first reversible early $\pi$-calculus,
and give it a denotational semantics in terms of reversible bundle event structures.
The new calculus is a reversible form of the internal $\pi$-calculus,
which is a subset of the $\pi$-calculus where every link sent by an output is private,
yielding greater symmetry between inputs and outputs.
\end{abstract}
\section{Introduction}\label{sec:intro}
The $\pi$-calculus~\cite{MPW92} is a widely used process calculus, which models communications between processes using input and output actions,
and allows the passing of communication links.
Various operational semantics of the $\pi$-calculus have been proposed, which can be classified according to whether transitions are unlabelled or labelled.
Unlabelled transitions (so-called reductions) represent completed interactions.
As observed in~\cite{SWU10} they give us the internal behaviour of complete systems,
whereas to reason compositionally about the behaviour of a system in terms of its components we need labelled transitions.
With labelled transitions, we can distinguish early and late semantics~\cite{MPW93},
with the difference being that early semantics allows a process to receive (free) names it already knows from the environment, while the late does not. This creates additional causation in the early case between those inputs and previous output actions making bound names free.
All existing reversible versions of the $\pi$-calculus use reduction semantics~\cite{lanese2010reversing,TIEZZI2015684} or late semantics~\cite{cristescu2013compositional,DBLP:journals/corr/abs-1808-08655}.
However the early semantics of the (forward-only) $\pi$-calculus is
more widely used than the late, partly because it has a sound 
correspondence with 
contextual congruences~\cite{10.5555/646246.684864,HondaKYoshida95}.

We define $\pi$IH, {the first reversible early $\pi$-calculus},
and give it a denotational semantics in terms of reversible event structures.
The new calculus is a reversible form of the internal $\pi$-calculus, or
$\pi$I-calculus~\cite{SANGIORGI1996235}, which is a subset of the $\pi$-calculus where every link sent by an output is bound (private),
yielding greater symmetry between inputs and outputs.
It has been shown that the asynchronous $\pi$-calculus can be encoded in the
asynchronous form of the $\pi$I-calculus~\cite{BOREALE1998205}.

The $\pi$-calculus has two forms of causation. \emph{Structural} causation, as one would find in CCS, comes directly from the structure of the process, e.g.\ in $\inp{a}{b}.\inp{c}{d}$ the action $\inp{a}{b}$ must happen before $\inp{c}{d}$.
\emph{Link} causation, on the other hand, comes from one action making a name available for others to use, e.g.\ in the process $\inp{a}{x}\vert \outp{b}{c}$, the event $\inp{a}{c}$ will be caused by $\outp{b}{c}$ making $c$ a free name.
Note that link causation as in this example is present in the early form of the $\pi$I-calculus though not the late{, since {it is created by the process receiving one of its free names}. Restricting ourselves to the $\pi$I-calculus, rather than the full $\pi$-calculus lets us focus on the link causation created by early semantics, since it removes the other forms of link causation present in the $\pi$-calculus.}


We base $\pi$IH on the work of Hildebrandt \emph{et al.}~\cite{hildebrandt2017stable}, which used extrusion histories and locations to define a stable non-interleaving early operational semantics for the $\pi$-calculus.
We extend the extrusion histories so that they contain enough information to reverse the $\pi$I-calculus, storing not only extrusions but also communications.
Allowing processes to evolve, while moving past actions to a history separate from the process, is called dynamic reversibility. 
By contrast, static reversibility, as in CCSK~\cite{PU07}, lets processes keep their structure during the computation, and annotations are used to keep track of the current state and how actions may be reversed.

Event structures are a model of concurrency which
describe causation, conflict and concurrency between events.
They are `truly concurrent' in that they do not reduce concurrency of events to the different possible interleavings.
They have been used to model forward-only process calculi~\cite{Crafa2012,Boudol1989,winskel1982event}, including the $\pi$I-calculus~\cite{Crafa2007compositional}. {Describing reversible processes as event structures is useful because it gives us a simple representation of the causal relationships between actions and gives us equivalences between processes which generate isomorphic event structures. True concurrency in semantics is particularly important in reversible process calculi, as the order actions can reverse in depends on their causal relations~\cite{journals/entcs/PhillipsU07}.}

Event structure semantics of dynamically reversible process calculi have the added complexity of the histories and the actions in the process being separated, obscuring the structural causation.
This was an issue for Cristescu \emph{et al.}~\cite{CristescuKV16}, who used rigid families~\cite{CastellanHLW14}, related to event structures, to describe the semantics of R$\pi$~\cite{cristescu2013compositional}.
Their semantics require a process to first reverse all actions to find the original process, map this process to a rigid family, and then apply each of the reversed memories in order to reach the current state of the process. Aubert and Cristescu~\cite{AUBERT201777} used a similar approach to describe the semantics of a subset of RCCS processes as configuration structures.
{We use a different tactic of first mapping to a statically reversible calculus, $\pi$IK, and then obtaining the event structure.} This means that while we do have to reconstruct the original structure of the process, we avoid redoing the actions in the event structure.

Our $\pi$IK is inspired by CCSK and the statically reversible $\pi$-calculus of~\cite{DBLP:journals/corr/abs-1808-08655}, which use communication keys to denote past actions.
To keep track of link causation,
keys are used in a number of different ways in~\cite{DBLP:journals/corr/abs-1808-08655}.
In our case we can handle link causation by using keys purely
to annotate the action which was performed using the key, and any names which were substituted during that action. 

Although our two reversible variants of the $\pi$I-calculus have very different syntax and originate from different ideas, we show an operational correspondence between them in Theorem~\ref{the:key-ext-eq}. We do this despite the extrusion histories containing more information than the keys, since they remember what bound names were before being substituted. 
The mapping from $\pi$IH to $\pi$IK bears some resemblance to the one presented from RCCS to CCSK in~\cite{Medic2016}, though with some important differences. $\pi$IH uses centralised extrusion histories more similar to rho$\pi$~\cite{LANESE201625} while RCCS uses distributed memories. Additionally, unlike CCS, $\pi$I has substitution as part of its transitions and memories are handled differently by $\pi$IK and $\pi$IH, and our mapping has to take this into account.

We describe denotational structural event structure semantics of $\pi$IK, partly inspired by~\cite{Crafa2012,Crafa2007compositional}, using reversible bundle event structures~\cite{EFG2018}. Reversible event structures~\cite{journals/jlp/PhillipsU15} allow their events to reverse and include relations describing when events can reverse. Bundle event structures are more expressive than prime event structures, since they allow an event to have multiple possible conflicting causes.
This allows us to model parallel composition without having a single action correspond to multiple events.
While it would be possible to model $\pi$IK using reversible prime event structures, using bundle event structures not only gives us fewer events, it also lays the foundation for adding rollback to $\pi$IK and $\pi$IH, similarly to~\cite{EFG2018}, which cannot be done using reversible prime event structures.

The structure of the paper is as follows: Section~\ref{sec:Ext-sem} describes $\pi$IH; Section~\ref{sec:piik} describes $\pi$IK; Section~\ref{sec:sem-corr} describes the mapping from $\pi$IH to $\pi$IK; Section~\ref{sec:BES} recalls labelled reversible bundle event structures; and Section~\ref{sec:Den-Ev-Sem} gives event structure semantics of $\pi$IK. Proofs of the results presented in this paper can be found in the technical report~\cite{graversen2020event}.

\section{$\pi$I-calculus reversible semantics with extrusion histories}\label{sec:Ext-sem}
Stable non-interleaving, early operational semantics of the $\pi$-calculus were defined by Hildebrandt \emph{et al.} in~\cite{hildebrandt2017stable}, using locations and extrusion histories to keep track of link causation. We will in this section use a similar approach to define a reversible variant of the $\pi$I-calculus, $\pi$IH, using the locations and histories to keep track of not just causation, but also past actions. The $\pi$I-calculus is a restricted variant of the $\pi$-calculus wherein output on a channel $a$, $\outp{a}{b}$, binds the name being sent, $b$, corresponding to the $\pi$-calculus process $(\nu b)\overline{a}\!\lrangles{b}\!.P$. This creates greater symmetry with the input $\inp{a}{x}$, where the variable $x$ is also bound. The syntax of $\pi$IH processes is:

\vspace{3pt}$P::=\sum\limits_{i\in I} \alpha_i.P_i \;\mid\; P_0\vert P_1 \;\mid \; (\nu x) P \;\;\;\; \alpha::=\outp{a}{b}\;\mid \; \inp{a}{b}$

\vspace{3pt}The forward semantics of $\pi$IH can be seen in Table~\ref{tab:ext-sem-fwd} and the reverse semantics can be seen in Table~\ref{tab:ext-sem-rev}. We associate each transition with an action $\mu::=\alpha\;\vert\; \tau$ and a location $u$ (Definition~\ref{def:Loc}), describing where the action came from and what changes are made to the process as a result of the action. We store these location and action pairs in extrusion and communication histories associated with processes, so $(\overline{H},\underline{H},H)\vdash\! P$ means that if $(\mu,u)$ is an action and location pair in the output history $\overline{H}$ then $\mu$ is an output action, which $P$ previously performed at location $u$. Similarly $\underline{H}$ contains pairs of input actions and locations and $H$ contains triples of two communicating actions and the location associated with their communication. We use $\mathbf{H}$ as shorthand for $(\overline{H},\underline{H},H)$.
\begin{definition}[Location \cite{hildebrandt2017stable}]\label{def:Loc}
A location $u$ of an action $\mu$ is one of the following:
\begin{enumerate}
\item $l[P][P']$ if $\mu$ is an input or output, where $l\in \{0,1\}^*$ describes the path taken through parallel compositions to get to $\mu$'s origin, $P$ is the subprocess reached by following the path before $\mu$ has been performed, and $P'$ is the result of performing $\mu$ in $P$.
\item $l\lrangles{0l_0[P_0][P_0'],1l_1[P_1][P_1']}$ if $\mu=\tau$, where $l0l_0[P_0][P_0']$ and $l1l_1[P_1][P_1']$ are the locations of the two actions communicating.
\end{enumerate}
The path $l$ can be empty if the action did not go through any parallel compositions.
\end{definition}

We also use the operations on extrusion histories from Definition~\ref{def:ExtOp}. These (1) add a branch to the path in every location, (2) isolate the extrusions whose locations begin with a specific branch, (3) isolate the extrusions whose locations begin with a specific branch and then remove the first branch from the locations, and (4) add a pair to the history it belongs in.
\begin{definition}[Operations on extrusion histories \cite{hildebrandt2017stable}]\label{def:ExtOp}
Given an extrusion history $(\overline{H},\underline{H},H)$, for $H^*\in\{\overline{H},\underline{H},H\}$ we have the following operations for $i\in \{0,1\}$:
\begin{enumerate}
\item $iH^*=\{(\mu,iu)\mid (\mu,u)\in H^*\}$
\item $[i]H^*=\{(\mu,iu)\mid (\mu,iu)\in H^*\}$
\item $[\check{i}]H^*=\{(\mu,u)\mid (\mu,iu)\in H^*\}$
\item $\mathbf{H}+(\mu,u)=\begin{dcases*}
(\overline{H}\cup \{L\},\underline{H},H) & if $(\mu,u)=(\outp{a}{n},u)$ \\
(\overline{H},\underline{H}\cup \{L\},H) & if $(\mu,u)=(\inp{a}{x},u)$ \\
(\overline{H},\underline{H},H\cup \{L\}) & if $(\mu,u)=(\inp{a}{x},\outp{a}{n},l\langle u_0,u_1\rangle)$ \\
\end{dcases*}$
\end{enumerate}
\end{definition}
\begin{table}[tb]
 \begin{tabular}{c}
\vspace{5pt}\infer[{[\text{OUT}]}]{{\mathbf{H}\vdash\!\displaystyle\sum\limits_{i\in I} \alpha_i.P_i} \,\xrightarrow[u]{\alpha_j} {(\overline{H}\cup \{(\outp{a}{n},u)\},\underline{H},H)\vdash\! P_j}}{u=[\sum\limits_{i\in I} \alpha_i.P_i][P_j]\;\;\;\alpha_j=\outp{a}{n}\;\;\; j\in I}\\
\vspace{5pt}\infer[{[\text{IN}]}]{{\mathbf{H}\vdash\!\sum\limits_{i\in I} \alpha_i.P_i} \,\xrightarrow[u]{\inp{a}{n}} {(\overline{H},\underline{H}\cup \{(\inp{a}{n},u)\},H)\vdash\! P_j'}}{u=[\sum\limits_{i\in I} \alpha_i.P_i][P_j]\;\;\;P_j'=P_j[x:=n]\;\;\;\alpha_j=\inp{a}{x}\;\;\; j\in I } \\

\vspace{8pt}\infer[{[\text{PAR}_i]}]{{\mathbf{H}\vdash\! P_0\vert P_1}\,\xrightarrow[iu]{\mu} {((\overline{H}\setminus [i]\overline{H})\cup i\overline{H'_i},(\underline{H}\setminus [i]\underline{H})\cup i\underline{H'_i},(H\setminus [i]H)\cup iH'_i)\vdash\! P_0'\vert P_1'}}{\begin{array}{l}
{([\check{i}]\overline{H},[\check{i}]\underline{H},[\check{i}]H)\vdash\! P_i} \,\xrightarrow[u]{\mu} {\mathbf{H}_i'\vdash\! P_i'}\;\;\; P_{1-i}'=P_{1-i}\;\;\;\text{if }\mu=\overline{a}(n)\text{ then } n\notin \fn(P_{1-i})
\end{array}} \\
\vspace{5pt}\infer[{[\text{COM}_i]}]
{{\mathbf{H}\vdash\! P_0\vert P_1} \,\xrightarrow[(0v_0,1v_1)]{\tau} {(\overline{H},\underline{H},H\cup\{((\alpha_0,\alpha_1,\langle 0v_0,1v_1\rangle)\})\vdash\! (\nu n)(P_0'\vert P_1')}}{\begin{array}{l}
{([\check{i}]\overline{H},[\check{i}]\underline{H},[\check{i}]H)\vdash\! P_i} \,\xrightarrow[v_i]{\alpha_i} {\mathbf{H}_i'\vdash\! P_i'} \;\;\; \alpha_i=\outp{a}{n} \;\;\; \alpha_j=\inp{a}{n} \vspace{3pt}\\ {([\check{j}]\overline{H},[\check{j}]\underline{H},[\check{j}]H)\vdash\! P_j} \,\xrightarrow[v_j]{\alpha_i} {\mathbf{H}_j'\vdash\! P_j'\;\;\; j=1-i\;\;\; n\notin \fn(P_j)}
\end{array}}\\
\vspace{5pt}\infer[{[\text{SCOPE}]}]{{\mathbf{H}\vdash\!(\nu x)P} \,\xrightarrow[u]{\mu} {\mathbf{H}'\vdash\!(\nu x)P'}}{{\mathbf{H}\vdash\! P} \,\xrightarrow[u]{\mu} {\mathbf{H}'\vdash\! P'}\;\; x\notin n(\mu)} \;\;\;\;\;\infer[{[\text{STR}]}]
{{\mathbf{H}\vdash\! P} \,\xrightarrow[u]{\mu} {\mathbf{H}'\vdash\! Q}}{P\equiv P'\;\;\; {\mathbf{H}\vdash\! P'} \,\xrightarrow[u]{\mu} {\mathbf{H}'\vdash\! Q'}\;\;\;Q'\equiv Q}\\
\end{tabular}
\caption{Semantics of $\pi$IH (forwards rules)}\label{tab:ext-sem-fwd}
\end{table}
The forwards semantics of $\pi$IH have six rules. In $[\text{OUT}]$ the action is an output, the location is the process before and after doing the output, and they are added to the output history. The equivalent reverse rule, $[\text{OUT}^{-1}]$, similarly removes the pair from the history and transforms the process from the second part of the location back to the first. The input rule $[\text{IN}]$ works similarly, but performs a substitution on the received name and adds the pair to the input history instead. In $[\text{PAR}_i]$ we isolate the parts of the histories whose locations start with $i$ and use those to perform an action in $P_i$, getting $\mathbf{H}_i'\vdash\! P_i'$. It then replaces the part of the histories parts of the histories whose locations start with $i$ with $\mathbf{H}_i'$ when propagating the action through the parallel. A communication in $[\text{COM}_i]$ adds memory of the communication to the history. The rules $[\text{SCOPE}]$ and $[\text{STR}]$ are standard and self-explanatory. 

The reverse rules use the extrusion histories to find a location $l[P][P']$ such that the current state of the subprocess at $l$ is $P'$, and change it to $P$.

In these semantics structural congruence, consisting only of $\alpha$-conversion together with ${!P}\equiv {!P\vert P}$ and ${(\nu~a)(\nu b)P}\equiv {(\nu~b)(\nu~a) P}$, is primarily used to create and remove extra copies of a replicated process when reversing the action that happened before the replication. Since we use locations in our extrusion histories, we try to avoid using structural congruence any more than necessary. However, not using it for parallel composition would mean that we would need some other way of preventing traces such as ${\mathbf{H}\vdash!P}\,\xrightarrow[u]{\mu}\xsquigarrow{\mu}{u}{\mathbf{H}\vdash!P\vert P}$, which allows a process to reach a state it could not reach via a parabolic trace. Using structural congruence for replication does not cause any problems for the locations, as we can tell past actions originating in each copy of $P$ apart by the path in their location, with actions from the $i$th copy having a path of $i$ $0$s followed by a $1$.

 \begin{table}[tb]
\small 
\begin{tabular}{c}
\vspace{5pt}
$\infer[{[\text{OUT}^{-1}]}]{\mathbf{H}\vdash P_j \xsquigarrow{\alpha_j}{u} (\overline{H}\setminus \left\{(\outp{a}{n},u)\right\},\underline{H},H)\vdash \sum_{i\in I} \alpha_i.P_i}{u=[\sum\limits_{i\in I} \alpha_i.P_i][P_j]\;\;\; \alpha_j=\outp{a}{n}\;\;\; j\in I\;\;\; (\outp{a}{n},u)\in \overline{H}} $ \\

 \vspace{5pt}
 $\infer[{[\text{IN}^{-1}]}]{\mathbf{H}\vdash P_j' \xsquigarrow{\inp{a}{n}}{u} (\overline{H},\underline{H}\setminus \left\{(\inp{a}{n},u)\right\},H)\vdash \sum\limits_{i\in I} \alpha_i.P_i}{u=[\sum\limits_{i\in I} \alpha_i.P_i][P_j]\;\;\;P_j'= P_j[x:=n]\;\;\;\alpha_j=\inp{a}{x}\;\;\; j\in I \;\;\; (\inp{a}{n},u)\in \underline{H}}$ \\
 \vspace{8pt}
 $\infer[{[\text{PAR}_i^{-1}]}]{\mathbf{H}\vdash P_0\vert P_1\xsquigarrow{\alpha}{iu} ((\overline{H}\setminus [i]\overline{H})\cup i\overline{H'_i},(\underline{H}\setminus [i]\underline{H})\cup i\underline{H'_i},(H\setminus [i]H)\cup iH'_i)\vdash P_0'\vert P_1'}{([\check{i}]\overline{H},[\check{i}]\underline{H},[\check{i}]H)\vdash P_i\xsquigarrow{\alpha}{u} \mathbf{H}_i'\vdash P_i'\;\;\; P_{1-i}'=P_{1-i} \text{ if }\alpha=\outp{a}{n}\text{ then } n\notin \fn(P_{1-i})}$\\
 \vspace{5pt}\infer[{[\text{COM}_i^{-1}]}]
{{\mathbf{H}\vdash\! (\nu n)(P_0\vert P_1)} \,\xsquigarrow{\tau}{(0v_0,1v_1)} {(\overline{H},\underline{H},H\setminus\{((\alpha_0,\alpha_1,\langle 0v_0,1v_1\rangle)\}\vdash\! P_0'\vert P_1'}}{\begin{array}{l}
([\check{i}]\overline{H}\cup \{(\outp{a}{n},v_i)\},[\check{i}]\underline{H},[\check{i}]H)\vdash P_i\xsquigarrow{\outp{a}{n}}{v_i} \mathbf{H}_i'\vdash P_i' \;\;\; \alpha_i=\outp{a}{n} \;\;\; \alpha_j=\inp{a}{n} \vspace{3pt}\\([\check{j}]\overline{H},[\check{j}]\underline{H}\cup \{(\inp{a}{n},v_j)\},[\check{j}]H)\vdash P_j\xsquigarrow{\inp{a}{n}}{v_j} \mathbf{H}_j'\vdash P_j'\;\;\; j=1-i\;\;\; n\notin \fn(P_j)
\end{array}}\\
 \vspace{5pt} 
$\infer[{[\text{SCOPE}^{-1}]}]{\mathbf{H}\vdash(\nu x)P\xsquigarrow{\mu}{u} \mathbf{H}'\vdash(\nu x)P'}{\mathbf{H}\vdash P\xsquigarrow{\mu}{u} \mathbf{H}'\vdash P'\;\; x\notin n(\alpha)}\;\;\;\;\;
\infer[{[\text{STR}^{-1}]}]{\mathbf{H}\vdash P \xsquigarrow{\alpha}{u} \mathbf{H}'\vdash Q}{P\equiv P'\;\;\; \mathbf{H}\vdash P' \xsquigarrow{\alpha}{u} \mathbf{H}'\vdash Q' \;\;\; Q'\equiv Q}$ \\
\end{tabular}
\caption{Semantics of reversible $\pi$IH (reverse rules)}\label{tab:ext-sem-rev}
\end{table}

\begin{example} 
Consider the process $(\inp{a}{x}.\outp{x}{d} \vert \outp{a}{c})\vert \inp{b}{y}$. If we start with empty histories, each transition adds actions and locations:

 \noindent{\resizebox{\textwidth}{!}{
 $\begin{array}{lr}
 {(\emptyset,\emptyset,\emptyset)\vdash\! (\inp{a}{x}.\outp{x}{d} \vert \outp{a}{c}) \vert \inp{b}{y}}& \hspace{-2cm}\xrightarrow[{0\langle 0[\inp{a}{x}.\outp{x}{d}][\outp{c}{d}],1[\outp{a}{c}][0]\rangle}]{\tau}\\
 {(\emptyset,\emptyset,\{(\inp{a}{c},\outp{a}{c},0\lrangles{ 0[\inp{a}{x}.\outp{x}{d}][\outp{c}{d}],1[\outp{a}{c}][0]}\})\vdash\! (\nu c)(\outp{c}{d} \vert 0)\vert \inp{b}{y}} & \xrightarrow[{00[\outp{c}{d}][0]}]{\outp{c}{d}}\\
 {(\{(\outp{c}{d},00[\outp{c}{d}][0])\},\emptyset,\{(\inp{a}{c},\outp{a}{c},0\lrangles{ 0[\inp{a}{x}.\outp{x}{d}][\outp{c}{d}],1[\outp{a}{c}][0]}\})\vdash\! (\nu c)(0\vert 0)\vert \inp{b}{y}}\; & \xrightarrow[{1[\inp{b}{y}][0]}]{\inp{b}{d}}\\
 \multicolumn{2}{l}{(\{(\outp{c}{b},00[\outp{c}{b}][0])\},\{(\inp{b}{d},{1[\inp{b}{y}][0]})\},\{(\inp{a}{c},\outp{a}{c},0\lrangles{ 0[\inp{a}{x}.\outp{x}{d}][\outp{c}{d}],1[\outp{a}{c}][0]}\})\vdash\! (0\vert 0)\vert 0}\\
 \end{array}$
 }}
 \end{example}

%
%

 We show that our forwards and reverse transitions correspond.

 \begin{proposition}[Loop]\label{prop:fwdtorevTrans}\mbox{}
 \begin{enumerate}
 \item Given a $\pi$IH process $P$ and an extrusion history $\mathbf{H}$, if ${\mathbf{H}\vdash\! P}\,\xrightarrow[u]{\alpha} {\mathbf{H}' \vdash\! Q}$, then ${\mathbf{H}'\vdash\! Q} \xsquigarrow{\alpha}{u} {\mathbf{H}\vdash\! P}$.
  \item Given a forwards-reachable $\pi$IH process $P$ and an extrusion history $\mathbf{H}$, if ${\mathbf{H}\vdash\! P} \xsquigarrow{\alpha}{u} {\mathbf{H}' \vdash\! Q}$, then ${\mathbf{H}'\vdash\! Q}\,\,\xrightarrow[u]{\alpha} {\mathbf{H}\vdash\! P}$.
  \end{enumerate}
  \end{proposition}

 
 \section{$\pi$I-calculus reversible semantics with annotations}\label{sec:piik}
In order to define event structure semantics of $\pi$IH, we first map from $\pi$IH to a statically reversible variant of $\pi$I-calculus, called $\pi$IK. $\pi$IK is based on previous statically reversible calculi $\pi$K~\cite{DBLP:journals/corr/abs-1808-08655} and
CCSK \cite{PU07}. Both of these use \emph{communication keys} to denote past actions and which other actions they have interacted with, so ${\inp{a}{x}\vert\outp{a}{b}}\xrightarrow{\tau[n]}{\inp{a}{b}[n]\vert\outp{a}{b}[n]}$ means a communication with the key $n$ has taken place between the two actions. We apply this idea to define early semantics of $\pi$IK, which has the following syntax: 

\vspace{3pt}
$P::= \alpha.P\,\mid\,\alpha[n].P\,\mid\, P_0+P_1\,\mid\, P_0\vert P_1\,\mid\, (\nu x)P \;\;\; \alpha::=\outp{a}{b}\,\mid \, \inp{a}{b}$

\vspace{3pt}The primary difference between applying communication keys to CCS and the $\pi$I-calculus is the need to deal with substitution. We need to keep track of not only which actions have communicated with each other, but also which names were substituted when. We do this by giving the substituted names a key, $a_{[n]}$, but otherwise treating them the same as those without the key, except when undoing the input associated with $n$.
  \begin{table}[tb] 
 \begin{center}
\begin{tabular}{c}
\infer{{\inp{a}{x}.P}\xrightarrow{\inp{a}{b}[n]} {\inp{a}{b}[n].P'}}{\mathsf{std}(P)\;\;\; P'=P[x:=b_{[n]}]}  \;\;\;\;\;  \infer{\outp{a}{b}.P\xrightarrow{\outp{a}{b}[n]} \outp{a}{b}[n].P}{\mathsf{std}(P)} \\ \infer{{\alpha[n].P}\xrightarrow{\mu[m]} {\alpha[n].P'}}{P\xrightarrow{\mu[m]} P' \;\;\; m\neq n\;\;\;\text{ if }\mu=\overline{a}(x) \text{ then } x\notin n(\alpha)} \;\;\;\; \infer{P_0 + P_1 \xrightarrow{\mu[n]} P_0'+ P_1}{P_0\xrightarrow{\mu[n]} P_0' \;\;\; \mathsf{std}(P_1)} \\
\infer{P_0\vert P_1\xrightarrow{\mu[n]} P_0'\vert P_1}{P_0\xrightarrow{\mu[n]} P_0' \;\;\; \mathsf{fsh}[n](P_1)\;\;\; \text{if }\mu=\outp{a}{b}\text{ then }b\notin \fn(P_1)} \;\;\;\;\; \infer{P_0\vert P_1\xrightarrow{\tau[n]} (\nu b)(P_0'\vert P_1')}{P_0\xrightarrow{\inp{a}{b}[n]}P_0'\;\;\; P_1\xrightarrow{\outp{a}{b}[n]} P_1'} \\
\vspace{5pt}\infer{(\nu a)P \xrightarrow{\mu[m]} (\nu a)P'}{P\xrightarrow{\mu[m]} P' \;\;\;\; a\notin n(\mu)} \;\;\;\;\;   \infer{P\xrightarrow{\mu[n]} P'}{P\equiv Q\xrightarrow{\mu[n]}Q'\equiv P'}\\
\end{tabular}
\caption{$\pi$IK forward semantics}\label{tab:pik-sem} \end{center}
 \end{table}
 
   \begin{table}[tb] 
 \begin{center}
\begin{tabular}{c}
\infer{\inp{a}{b}[m].P\xrsquigarrow{\inp{a}{b}[m]} \inp{a}{x}.P'}{\mathsf{std}(P)\;\;\; x\notin n(P) \;\;\; P'=P[b_{[m]}:=x]} \;\;\;\; \infer{\outp{a}{b}[n].P\xrsquigarrow{\outp{a}{b}[n]} \outp{a}{b}.P}{\mathsf{std}(P)} \\ 
\infer{\alpha[n].P\xrsquigarrow{\mu[m]} \alpha[n].P'}{P\xrsquigarrow{\mu[m]} P' \;\;\; m\neq n} \;\;\;\; \infer{P_0 + P_1 \xrsquigarrow{\mu[n]} P_0'+ P_1}{P_0\xrsquigarrow{\mu[n]} P_0' \;\;\; \mathsf{std}(P_1)}\\
\infer{P_0\vert P_1\xrsquigarrow{\mu[n]} P_0'\vert P_1}{P_0\xrsquigarrow{\mu[n]} P_0' \;\;\; \mathsf{fsh}[n](P_1)\;\;\; \text{ if } \mu=\outp{a}{b} \text{ then } b\notin \fn(P_1)} \;\;\;\;\; \infer{(\nu b)(P_0\vert P_1)\xrsquigarrow{\tau[n]} P_0'\vert P_1'}{P_0\xrsquigarrow{\inp{a}{b}[n]}P_0' \;\;\; P_1\xrsquigarrow{\outp{a}{b}[n]} P_1'} \\
\vspace{5pt} \infer{(\nu a)P \xrsquigarrow{\mu[m]} (\nu a)P'}{P\xrsquigarrow{\mu[m]} P' \;\;\; a\notin n(\mu)} \;\;\;\;\; \infer{P\xrsquigarrow{\mu[n]} P'}{P\equiv Q\xrsquigarrow{\mu[n]}Q'\equiv P'}\\
\end{tabular}
\caption{$\pi$IK reverse semantics}\label{tab:pik-rev-sem} \end{center}
 \end{table}
 
Table \ref{tab:pik-sem} shows the forward semantics of $\pi$IK. The reverse semantics can be seen in Table~\ref{tab:pik-rev-sem}. We use~$\alpha$ to range over input and output actions and $\mu$ over input, output, and~$\tau$. We use $\mathsf{std}(P)$ denote that~$P$ is a \emph{standard process}, meaning it does not contain any past actions (actions annotated with a key), and $\mathsf{fsh}[n](P)$ to denote that a key~$n$ is fresh for~$P$. Names in past actions are always free. Our semantics very much resemble those of CCSK, with the exceptions of substitution and ensuring that any name being output does not appear elsewhere in the process. 
The semantics use structural congruence as defined in Table~\ref{tab:str-con}.
\begin{table}[t]
\begin{tabular}{lll}
		$P\vert 0\equiv P$ & $P_0\vert P_1\equiv P_1\vert P_0$ & $P_0\vert(P_1\vert P_2)\equiv (P_0\vert P_1)\vert P_2$ \\
		$P+0\equiv P$ \hspace{1cm} & $P_0+P_1\equiv P_1+P_0$ \hspace{1cm} & $P_0+(P_1+P_2)\equiv (P_0+P_1)+P_2$ \\
\vspace{5pt}		$!P\equiv {!P\vert P} $ & $(\nu x)(\nu y)P\equiv (\nu y)(\nu x) P$ & $(\nu a)(P_0\vert P_1) \equiv ((\nu a) P_0\vert P_1)$ if $a\notin n(P_1)$\\
	\end{tabular}
	\caption{Structural congruence}\label{tab:str-con}
\end{table}	
	

We again show a correspondence between forward and reverse transitions.
\begin{proposition}[Loop]\label{the:FwdToRev}\mbox{}
\begin{enumerate}
\item Given a process $P$, if $P\xrightarrow{\mu[n]} Q$ then $Q\xrsquigarrow{\mu[n]} P$.
\item Given a forwards reachable process $P$, if $P\xrsquigarrow{\mu[n]} Q$ then $Q\xrightarrow{\mu[n]} P$.
\end{enumerate} 
\end{proposition}




 \section{Mapping from $\pi$IH to $\pi$IK}\label{sec:sem-corr}
We will now define a mapping from $\pi$IH to $\pi$IK and show that we have an operational correspondence in Theorem~\ref{the:key-ext-eq}. The extrusion histories store more information than the keys, as they keep track of which names were substituted, as illustrated by Example~\ref{ex:ext-key-sub}. This means we lose some information in our mapping, but not information we need.

\begin{example}\label{ex:ext-key-sub}
Consider the processes $(\emptyset,\{(a(b),[a(x)][0])\},\emptyset)\vdash\! 0$ and $a(b)[n]$. These are the result of $a(x)$ receiving $b$ in the two different semantics. We can see that the extrusion history remembers that the input name was $x$ before $b$ was received, but the keys do not remember, and when reversing the action could use any name as the input name. This does not make a great deal of difference, as after reversing $a(b)$, the process with the extrusion history can also $\alpha$-convert $x$ to any name.
\end{example}

Since we intend to define a mapping from processes with extrusion histories to processes with keys, we first describe how to add keys to substituted names in a process in Definition~\ref{def:S}. We have a function, $S$, which takes a process, $P_1$, in which we wish to add the key $[n]$ to all those names which were $x$ in a previous state of the process, $P_2$, before being substituted for some other name in an input action with the key $[n]$.

\begin{definition}[Substituting in $\pi$IK-process to correspond with processes with extrusion histories]\label{def:S} Given a $\pi$IK process $P_1$, a $\pi$I-calculus process without keys, $P_2$, a key $n$, and a name $x$, we can add the key $n$ to any names which $x$ has been substituted with, by applying $S(P_1,P_2,[n],x)$, defined as:

\begin{enumerate}
\item $S\left(0,0,[n],x\right)=0$  \vspace{5pt}
\item $S\left(\sum\limits_{i\in I} P_{i1},\sum\limits_{i\in I} P_{i2},[n],x\right)=\sum\limits_{i\in I}S\left(P_{i1},P_{i2},[n],x\right)$ \vspace{5pt}
\item $S\left(P_1\vert Q_1,P_2\vert Q_2,[n],x\right)=S\left(P_1,P_2,[n],x\right)\vert S\left(Q_1,Q_2,[n],x\right)$ \vspace{5pt}
\item $S\left((\nu a) P_1, (\nu b) P_2,[n],x\right)= P_1'$ where: \\
\ind if $x=b$ then $P_1'=P_1$ and otherwise $P_1'=(\nu a) S\left(P_1,P_2,[n],x\right)$. \vspace{5pt}
\item $S\left(\alpha_1.P_1,\alpha_2.P_2,[n],x\right)=\alpha_1'.P_1'$ where: \\
\ind if $\alpha_2\in \{\inp{x}{c},\outp{x}{c}\}$ then $\alpha_1'=\alpha_{1_{[n]}}$ and otherwise $\alpha_1'=\alpha_1$; \\
\ind if $\alpha_2\in \{\inp{c}{x},\outp{c}{x}\}$ then $P_1'=P_1$ and otherwise $P_1'=S\left(P_1,P_2,[n],x\right)$. \vspace{5pt}
\item $S\left(\alpha_1[m].P_1,\alpha_2.P_2,[n],x\right)=\alpha_1'[m].P_1'$ where: \\
\ind if $\alpha_2\in \{\inp{x}{c},\outp{x}{c}\}$ then $\alpha_1'=\alpha_{1_{[n]}}$ and otherwise $\alpha_1'=\alpha_1$; \\
\ind if $\alpha_2\in \{\inp{c}{x},\outp{c}{x}\}$ then $P_1'=P_1$ and otherwise $P_1'=S\left(P_1,P_2,[n],x\right)$.\vspace{5pt}
\item $S\left(!P_1, !P_2,[n],x\right)= {!S\left(P_1,P_2,[n],x\right)}$ \vspace{5pt}
\item $S\left(P_1\vert P_1', !P_2,[n],x\right)= S\left(P_1,!P_2,[n],x\right)\vert S\left(P_1',P_2,[n],x\right)$ \vspace{5pt}
\item $S\left(!P_1,P_2\vert P_2', [n],x\right)=S\left(!P_1,P_2,[n],x\right)\vert S\left(P_1,P_2',[n],x\right)$
\end{enumerate}

\noindent where $\inp{a}{b}_{[n]}=\inp{a_{[n]}}{b}$ and $\outp{a}{b}_{[n]}=\outp{a_{[n]}}{b}$ 
\end{definition}

Being able to annotate our names with keys, we can define a mapping, $E$, from extrusion histories to keys in Definition~\ref{def:E}. $E$ iterates over the extrusions, having one process which builds $\pi$IK-process, and another that keeps track of which state of the original $\pi$IH process has been reached. When turning an extrusion into a keyed action, we use the locations as key and also give each extrusion an extra copy of its location to use for determining where the action came from. This way we can use one copy to iteratively go through the process, removing splits from the path as we go through them, while still having another intact copy of the location to use as the final key. {In $E(\mathbf{H}\vdash\! P,P')$, $\mathbf{H}$ is a history of extrusions which need to be turned into keyed actions, $P$ is the process these keyed actions should be added to, and $P'$ is the state the process would have reached, had the added extrusions been reversed instead of turned into keyed actions.}

If $E$ encounters a parallel composition in $P$ (case 2), it splits its extrusion histories in three. One part, $\mathbf{H}_{\mathsf{shared}}$ contains the locations which have an empty path, and therefore belong to actions from before the processes split. Another part contains the locations beginning with $0$, and goes to the first part of the process. And finally the third part contains the locations beginning with $1$, and goes to the second part of the process. 

{$E$ can add an action -- and the choices not picked when that action was performed -- to $P$ (cases 3,4) when the associated location has an empty path and has $P'$ as its result process.} 
When turning an input memory from the history into a past input action in the process (case 4), we use $S$ (Definition~\ref{def:S}) to add keys to the substituted names. When $E$ encounters a restriction (case 5), it moves a memory that can be used inside the restriction inside. It does this iteratively until there are no such memories left in the extrusion histories. We apply $E$ to a process in Example~\ref{ex:E}.
\begin{definition}
The function $\mathsf{lcopy}$ gives each member of an extrusion history an extra copy of its location:

\begin{tabular}{l}
$\mathsf{lcopy}(H^*)=\{(\mu,u,u)\mid (\mu,u)\in H^*\}$ \\ $\mathsf{lcopy}(\overline{H},\underline{H},H)=(\mathsf{lcopy}(\overline{H}),\mathsf{lcopy}(\underline{H}),\mathsf{lcopy}(H))$\\
\end{tabular}
\end{definition}

\begin{definition}\label{def:E}
Given a $\pi$IH process, $\mathbf{H}\vdash\! P$, we can create an equivalent $\pi$IK process, $E(\mathsf{lcopy}(\mathbf{H})\vdash\! P,P)=P'$ defined as 
\begin{enumerate}
\item $E((\emptyset,\emptyset,\emptyset)\vdash\! P,P')=P$


\vspace{5pt}
\item $E(\mathbf{H}\vdash\! P_0\vert P_1,P_0'\vert P_1')=E(\mathbf{H}_{\mathsf{shared}}\vdash\! P_0'' \vert P_1'' ,P_0'''\vert P_1''')$ where:

$\;\;\; \mathbf{H}_{\mathsf{shared}}=(\{(\alpha,u,u')\mid {(\alpha,u,u')\in \overline{H}} \text{ and }u\neq iu''\},\{(\alpha,u,u')\mid {(\alpha,u,u')\in \underline{H}}$ 

$\;\;\; {\text{and }} {u\neq iu''}\},\emptyset)$

$\ind P_0''=E((\overline{H_0},\underline{H_0},H_0)\vdash\! P_0,P_0')$ where:

$\ind\ind \overline{H_0}=\{(\outp{a}{b},u_0,u_0')\mid (\outp{a}{b},0u_0,u_0')\in \overline{H} \text{ or }{(\outp{a}{b},\alpha_1,\lrangles{0u_0,1u_1},u_0')}\in H\}$

$\ind\ind \underline{H_0}=\{ (\inp{a}{b},u_0,u_0')\mid (\inp{a}{b},0u_0,u_0')\in \underline{H} \text{ or }{(\inp{a}{b},\alpha_1,\lrangles{0u_0,1u_1},u_0')}\in H\}$

$\ind\ind H_0=\{(\alpha,\alpha',u,u')\mid (\alpha,\alpha',0u,u')\in H\}$

$\ind P_1''=E((\overline{H_1},\underline{H_1},H_1)\vdash\! P_1,P_1'))$ where:

$\ind\ind \overline{H_1}=\{(\outp{a}{b},u_1,u_1')\mid (\outp{a}{b},1u_1,u_1')\in \overline{H} \text{ or }{(\alpha_0,\outp{a}{b},\lrangles{0u_0,1u_1},u_1')}\in H\}$

$\ind\ind \underline{H_1}=\{(\inp{a}{b},u_1,u_1')\mid (\inp{a}{b},1u_1,u_1')\in \underline{H} \text{ or }{(\alpha_0,\inp{a}{b},\lrangles{0u_0,1u_1},u_1')}\in H\}$

$\ind\ind H_1=\{(\alpha,\alpha',u,u')\mid (\alpha,\alpha',1u,u')\in H\}$

$\mathbf{H_i}\vdash\! P_i' \xsquigarrow{\alpha_{i,0}}{u_{i,0}} \dots \xsquigarrow{\alpha_{i,n}}{u_{i,n}} (\emptyset,\emptyset,\emptyset)\vdash\! P_i'''$ for $i\in \{0,1\}$
\vspace{5pt}
\item $E((\overline{H}\cup \{(\outp{a}{b},[Q][P'],u)\},\underline{H},H)\vdash\! P,P')=E(\mathbf{H}\vdash\! \outp{a}{b}\left[u\right].P+\sum\limits_{i\in I\setminus \{j\}} \alpha_i.P_i,Q)$ 

$\ind$ if $Q=\sum_{i\in I} \alpha_i.P_i$, $\outp{a}{b}=\alpha_j$, and $P'=P_j$
\vspace{5pt}
\item $E((\overline{H},\underline{H}\cup\{(\inp{a}{b},[Q][P'],u)\},H)\vdash\! P,P')=$\\ \hspace*{\fill} $E(\mathbf{H}\vdash\! \inp{a}{b}\left[u\right].S(P,P_j,[u],x)+\sum\limits_{i\in I\setminus \{j\}} \alpha_i.P_i,Q)$ 

$\ind$ if $Q=\sum_{i\in I} \alpha_i.P_i$, $\inp{a}{x}=\alpha_j$, and $P'=P_j[x:=b]$
\vspace{5pt}
\item $E(\mathbf{H}\vdash\! (\nu x)P, (\nu x)P')=E(\mathbf{H}-(\alpha,u,u')\vdash\! P'',(\nu x)Q')$ 

$\ind$ where $P''=(\nu x)E((\emptyset,\emptyset,\emptyset)+(\alpha,u,u')\vdash\! P,P')$ 

$\ind$ if $(\alpha,u,u')\in \overline{H}\cup\underline{H}$ and $(\emptyset,\emptyset,\emptyset)+(\alpha,u,u)\vdash\! P \xsquigarrow{\alpha}{u} (\emptyset,\emptyset,\emptyset)\vdash\! Q'$
\vspace{5pt}
\item $E(\mathbf{H}\vdash\! !P,!P')=E(\mathbf{H}\vdash\! !P\vert P,!P'\vert P')$ if there exists $(\alpha,u,u')\in \overline{H}\cup\underline{H}\cup H$ such that $u\neq [Q][Q']$.
\end{enumerate}
\end{definition}

\begin{example}\label{ex:E}
We will now apply $E$ to the process $$(\{(\outp{b}{c},u_2)\},\emptyset,\{(\inp{b}{a},\outp{b}{a},\langle{0u_0,1u_1}\rangle )\})\vdash\! \inp{a}{x}\mid 0$$ with locations $u_0=[\inp{b}{y}.\inp{y}{x}][\inp{a}{x}]$, $u_1=[\outp{b}{a}][0]$, and $u_2=[\outp{b}{c}.(\inp{b}{y}.\inp{y}{x}\mid\outp{b}{a}][\inp{b}{y}.\inp{y}{x}\mid\outp{b}{a}]$. We perform $$E(\mathsf{lcopy}((\{(\outp{b}{c},u_2)\},\emptyset,\{(\inp{b}{a},\outp{b}{a},\langle{0u_0,1u_1}\rangle )\}))\vdash \inp{a}{x}\mid 0,\inp{a}{x}\mid 0)$$

Since we are at a parallel, we use Case 2 of Definition~\ref{def:E} to split the extrusion histories into three to get $E((\{(\outp{b}{c},u_2,u_2)\},\emptyset,\emptyset)\vdash\!P_0\mid P_1,\inp{b}{y}.\inp{y}{x}\mid\outp{b}{a})$ where $P_0=E((\emptyset,\{(\inp{b}{a},u_0,\langle{0u_0,1u_1}\rangle)\},\emptyset)\vdash\! \inp{a}{x},\inp{a}{x})$ and $P_1= E((\{(\outp{b}{a},u_1,\langle{0u_0,1u_1}\rangle)\},\emptyset,\emptyset)\vdash\! 0,0)$. 

To find $P_0$, we look at $u_0$, and find that it has $\inp{a}{x}$ as its result, meaning we can apply Case 4 to obtain $E((\emptyset,\emptyset,\emptyset)\vdash\! \inp{b}{a}[\langle{0u_0,1u_1}\rangle].S(\inp{a}{x},\inp{y}{x},[\langle{0u_0,1u_1}\rangle],y), \inp{b}{y}.\inp{y}{x})$.
And by applying Case 5 of Definition~\ref{def:S}, $S(\inp{a}{x},\inp{y}{x},[\langle{0u_0,1u_1}\rangle],y)=\inp{a_{{[\langle{0u_0,1u_1}\rangle]}}}{x}$. Since we have no more extrusions to add, we apply Case 1 to get our process $P_0=\inp{b}{a}[\langle{0u_0,1u_1}\rangle].\inp{a_{{[\langle{0u_0,1u_1}\rangle]}}}{x}$.

To find $P_1$, we similarly look at $u_1$ and find that we can apply Case 3. This gives us $P_1=\outp{b}{a}[\langle{0u_0,1u_1}\rangle].0$.

We can then apply Case 3 to $E((\{(\outp{b}{c},u_2,u_2)\},\emptyset,\emptyset)\vdash\! P_0 \mid P_1,\inp{b}{y}.\inp{y}{x}\mid\outp{b}{a})$. This gives us our final process, $$\outp{b}{c}[k']. \inp{b}{a}[k].\inp{a_{{[k]}}}{x} \mid \outp{b}{a}[k].0$$ where $k=\langle{0u_0,1u_1}\rangle$ and $k'=u_2$
\end{example}
We can then show, in Theorem~\ref{the:key-ext-eq}, that we have an operational correspondence between our two calculi and $E$ preserves transitions. Item 1 states that every transition in $\pi$IH corresponds to one in $\pi$IK process generated by $E$, and Item 2 vice versa. 

\begin{theorem}\label{the:key-ext-eq}
Given a reachable $\pi$IH process, $\mathbf{H}\vdash\! P$, and an action, $\mu$, 
\begin{enumerate}
\item if there exists a location $u$ such that $\mathbf{H}\vdash\! P\xarrowtail{\mu}{u} \mathbf{H}'\vdash\! P'$ then there exists a key, $m$, such that $E(\mathsf{lcopy}(\mathbf{H})\vdash\! P,P)\xarrowtail{\mu[m]}{} E(\mathsf{lcopy}(\mathbf{H}')\vdash\! P',P')$;
\item if there exists a key, $m$, such that $E(\mathsf{lcopy}(\mathbf{H})\vdash\! P,P)\xarrowtail{\mu[m]}{}P''$, then there exists a location, $u$, and a $\pi$IH process, $\mathbf{H}'\vdash\! P'$, such that $\mathbf{H}\vdash\! P\xarrowtail{\mu}{u} \mathbf{H}'\vdash\! P'$ and $P''\equiv E(\mathsf{lcopy}(\mathbf{H}')\vdash\! P',P')$.
\end{enumerate}
\end{theorem}

\section{Bundle event structures}\label{sec:BES}
 In this section we will recall the definition of \emph{labelled reversible bundle event structures} (LRBESs), which we intend to use later to define the event structure semantics of $\pi$IK and through that $\pi$IH. We also describe some operations on LRBESs, which our semantics will make use of. This section is primarily a review of definitions from~\cite{EFG2018}. We use bundle event structures, rather than the more common prime event structures, because LRBESs yield more compact event structures with fewer events and simplifies parallel composition. 
 
 An LRBES consists of a set of events, $E$, a subset of which, $F$, are reversible,and three relations on them. The bundle relation, $\mapsto$, says that if $X\mapsto e$ then one of the events of $X$ must have happened before $e$ can and all events in $X$ are in conflict with each other. The conflict relation, $\cf$, says that if $e\cf e'$ then $e$ and $e'$ cannot occur in the same configuration. The prevention relation, $\rhd$, says that if $e\rhd \underline{e'}$ then $e'$ cannot reverse after $e$ has happened. Since the event structure is labelled, we also have a set of labels $\mathsf{Act}$, and a labelling function $\lambda$ from events to labels. We use $\underline{e}$ to denote $e$ being reversed, and $e^*$ to denote either $e$ or $\underline{e}$.
 
 	\begin{definition}[Labelled Reversible Bundle Event Structure \cite{EFG2018}] 
		A labelled reversible bundle event structure is a 7-tuple $\M{E}=(E,F,\mapsto,\cf ,\rhd,\lambda,\mathsf{Act})$ where:
		\begin{enumerate}
			\item $E$ is the set of events;
			\item $F\subseteq E$ is the set of reversible events;
			\item the bundle set, ${\mapsto}\subseteq 2^E\times (E\cup \underline{F})$, satisfies $X\mapsto e^*\Rightarrow \forall e_1,e_2\in X.e_1\neq e_2\Rightarrow e_1\cf e_2$ and for all $e\in F$, $\{e\}\mapsto \underline{e}$;
			\item the conflict relation, ${\cf} \subseteq E\times E$, is symmetric and irreflexive;
			\item $\rhd \subseteq E\times \underline{F}$ is the prevention relation.
			\item $\lambda:E\rightarrow \mathsf{Act}$ is a labelling function.
		\end{enumerate}
	\end{definition}
	
%
	
	An event in a LRBES can have multiple possible causes as defined in Definition~\ref{def:Pos-Cause}. A possible cause $X$ of an event $e$ is a conflict-free set of events which contains a member of each bundle associated with $e$ and contains possible causes of all events in $X$.
	\begin{definition}[Possible Cause]\label{def:Pos-Cause}
Given an LRBES, $\M{E}=\LRBES{}$ and an event $e\in E$, $X\subseteq E$ is a possible cause of $e$ if 
\begin{itemize}
\item $e\notin X$, $X \text{is finite}$, whenever $X'\mapsto e$ we have $X'\cap X\neq \emptyset$;
\item for any $ e',e''\in\{e\}\cup X$, we have $e'\not\cf e''$ ($X\cup \{e\}$ is conflict-free);
\item for all $e'\in X$, there exists $X''\subseteq X$, such that $X''$ is a possible cause of $e'$;
\item there does not exist any $X'''\subset X$, such that $X'''$ is a possible cause of $e$.
\end{itemize}
\end{definition}
	
Since we want to compare the event structures generated by a process to the operational semantics, we need a notion of transitions on event structures. For this purpose we use configuration systems (CSs), which event structures can be translated into.

	\begin{definition}[Configuration system \cite{journals/jlp/PhillipsU15}]\label{def:CS}
		A \emph{configuration system} (CS) is a quadruple $\M{C} = \CS{}$ where $E$ is a set of events, $F\subseteq E$ is a set of reversible events, $\mathsf{C}\subseteq 2^E$ is the set of configurations, and $\rightarrow\subseteq \mathsf{C}\times 2^{E\cup\underline{F}} \times \mathsf{C}$ is a labelled transition relation such that if $X\CStrans{A}{B} Y$ then:
		\begin{itemize}
			\item $X,Y\in \mathsf{C}$, $A\cap X=\emptyset$; $B\subseteq X\cap F$; and $Y=(X\setminus B)\cup A$;
			\item for all $A'\subseteq A$ and $B'\subseteq B$, we have $X\CStrans{A'}{B'} Z \CStrans{(A\setminus A')}{(B\setminus B')}Y$, meaning $Z=(X\setminus B')\cup A'\in \mathsf{C}$.
		\end{itemize}
	\end{definition}
	\begin{definition}[From LRBES to CS \cite{EFG2018}]\label{def:RBEStoCS}
		We define a mapping $C_{br}$ from LRBESs to CSs as:
	 $C_{br}(\LRBES{})=\CS{}$ where:
			\begin{enumerate}
				\item $X\in\textsf{C}$ if $X$ is conflict-free;
				\item \vspace{-.2cm}For $X,Y\in \textsf{C}$, $A\subseteq E$, and $B\subseteq F$, there exists a transition $X\CStrans{A}{B} Y$ if:
				\begin{enumerate}
					\item $Y=(X\setminus B)\cup A$; $X\cap A=\emptyset$; $B\subseteq X$; and $X\cup A$ conflict-free;
					\item for all $e\in B$, if $e'\rhd \underline{e}$ then $e'\notin X\cup A$;
					\item for all $e\in A$ and $X'\subseteq E$, if $X'\mapsto e$ then $X'\cap (X\setminus B)\neq \emptyset$;
					\item for all $e\in B$ and $X'\subseteq E$, if $X'\mapsto \underline{e}$ then $X'\cap (X\setminus (B\setminus \{e\}))\neq \emptyset$.
				\end{enumerate}
			\end{enumerate}
	\end{definition}
For our semantics we need to define a prefix, restriction, parallel composition, and choice. Causal prefixing takes a label, $\mu$, an event, $e$, and an LRBES, $\M{E}$, and adds $e$ to $\M{E}$ with the label $\mu$ and associating every other event in $\M{E}$ with a bundle containing only $e$. Restriction removes a set of events from an LRBES.
\begin{definition}[Causal Prefixes \cite{EFG2018}]\label{def:CausPref}
Given an LRBES $\M{E}$, a label $\mu$, and an event $e$, $(\mu)(e).\M{E}=(E',F',\mapsto',\cf',\rhd',\lambda',\mathsf{Act}')$ where:
\vspace{-5pt}
\begin{multicols}{2}
\begin{enumerate}
\item $E'=E\cup e$
\item $F'=F\cup e$
\item ${\mapsto'}={\mapsto\cup (\{\{e\}\}\times (E\cup \{\underline{e}\}))}$
\item ${\cf'}={\cf}$
\item $\rhd'=\rhd \cup (E\times \{ \underline{e} \})$
\item $\lambda'=\lambda[e\mapsto \mu]$
\item $\mathsf{Act}'=\mathsf{Act} \cup \{ \mu \}$
\end{enumerate}
\end{multicols}
\end{definition} 

%
Removing a set of labels $L$ from an LRBES removes not just events with labels in~$A$ but also events dependent on events with labels in $L$.

\begin{definition}[Removing labels and their dependants]
Given an event structure $\M{E}=\LRBES{}$ and a set of labels $L\subseteq \mathsf{Act}$, we define $\rho_{\M{E}}(L)=X$ as the maximum subset of $E$ such that 
\begin{enumerate}
\item if $e\in X$ then $\lambda(e)\notin L$;
\item if $e\in X$ then there exists a possible cause of $e$, $x$, such that $x\subseteq X$.
\end{enumerate}
\end{definition}
A choice between LRBESs puts all the events of one event structure in conflict with the events of the others.
\begin{definition}[Choice \cite{EFG2018}]
Given LRBESs $\M{E}_0,\M{E}_1,\dots,\M{E}_n$, the choice between them is $\sum\limits_{0\leq i\leq n}\M{E}_i=\LRBES{}$ where:
\vspace{-5pt}
\begin{multicols}{2}
\begin{enumerate}
\item $E=\bigcup\limits_{0\leq i\leq n} \{i\}\times E_i$
\item $F=\bigcup\limits_{0\leq i\leq n} \{i\}\times F_i$
\item $X\mapsto e^*$ if $e=(i,e_i)$, $X_i\mapsto_i e_i^*$, and $X=\{i\}\times X_i$
\item $(i,e)\cf (j,e')$ if $i\neq j$ or $e\cf_i e'$
\item $(i,e)\rhd (j,e')$ if $i\neq j$ or $e\cf_i e'$
\item $\lambda(j,e)=\lambda_j(e)$
\item $\mathsf{Act}=\bigcup\limits_{0\leq i\leq n} \mathsf{Act}_i$
\end{enumerate}
\end{multicols}
\end{definition}

 \begin{definition}[Restriction \cite{EFG2018}]
Given an LRBES, $\M{E}=\LRBES{}$, restricting $\M{E}$ to $E'\subseteq E$ creates $\M{E}\upharpoonright E'=(E',F',\mapsto',\cf',\rhd',\lambda',\mathsf{Act}')$ where:
\vspace{-5pt}
\begin{multicols}{2}
\begin{enumerate}
\item $F'=F\cap E'$;
\item ${\mapsto'}={\mapsto\cap(\mathcal{P}(E')\times(E'\cup\underline{F'}))}$;
\item ${\cf'}={\cf\cap(E'\times E')}$;
\item ${\rhd'}={\rhd\cap (E'\times\underline{F'})}$;
\item $\lambda'=\lambda\upharpoonright_{E'}$;
\item $\mathsf{Act}=\mathsf{ran}(\lambda')$.
\end{enumerate} 
\end{multicols}
\end{definition}

 For parallel composition we construct a product of event structures, which consists of events corresponding to synchronisations between the two event structures. The possible causes of an event $(e_0,e_1)$ contain a possible cause of $e_0$ and a possible cause of $e_1$.
	\begin{definition}[Parallel \cite{EFG2018}]\label{def:RBESpro} 
		Given two LRBESs $\M{E}_0=\LRBES{0}$ and $\M{E}_1=\LRBES{1}$, their parallel composition $\M{E}_0\times \M{E}_1=\LRBES{}$ with projections $\pi_0$ and $\pi_1$ where:
		\begin{enumerate}
			\item $E=E_0\times_* E_1=\{(e,*) \mid e\in E_0\} \cup \{(*,e) \mid e\in E_1\} \cup \{(e,e') \mid e\in E_0\text{ and } e'\in E_1\}$;
			\item $F=F_0\times_* F_1=\{(e,*) \mid e\in F_0\} \cup \{(*,e) \mid e\in F_1\} \cup \{(e,e') \mid e\in F_0\text{ and } e'\in F_1\}$;
			\item
			for $i \in \{0,1\}$ we have
			$(e_0,e_1)\in E$, $\pi_i((e_0,e_1))=e_i$;
			\item for any $e^*\in E\cup \underline{F}$, $X\subseteq E$, $X\mapsto e^*$ iff there exists $i\in \{0,1\}$ and $X_i\subseteq E_i$ such that $X_i\mapsto \pi_i(e)^*$ and $X=\{e'\in E\mid \pi_i(e')\in X_i\}$;
			\item for any $e,e'\in E$, $e\cf e'$ iff there exists $i\in \{0,1\}$ such that $\pi_i(e)\cf_i\pi_i(e')$, or $\pi_i(e)=\pi_i(e')\neq \bot$ and $\pi_{1-i}(e)\neq \pi_{1-i}(e')$;
			\item for any $e\in E$, $e'\in F$, $e\rhd \underline{e'}$ iff there exists $i\in \{0,1\}$ such that $\pi_i(e)\rhd_i\underline{\pi_i(e')}$.
			
 \item $\lambda(e)=\begin{dcases*}
 \lambda_0(e_0) & if $e=(e_0,*)$\\
 \lambda_1(e_1) & if $e=(*,e_1)$\\
\tau & if $e=(e_0,e_1)$ and either $\lambda_0(e_0)=\inp{a}{x}$ and $\lambda_1(e_1)=\outp{a}{x}$\\ & or $\lambda_0(e_0)=\outp{a}{x}$ and $\lambda_1(e_1)=\inp{a}{x}$\\
0 & otherwise
 \end{dcases*}$ \vspace{5pt}
 \item $\mathsf{Act}=\{\tau\}\cup \mathsf{Act}_0\cup \mathsf{Act}_1$
		\end{enumerate}
	\end{definition}
 \section{Event structure semantics of $\pi$IK}\label{sec:Den-Ev-Sem} 
 In this section we define event structure semantics of $\pi$IK using the LRBESs and operations defined in Section~\ref{sec:BES}. Theorems~\ref{the:PtoLRBEStrans} and~\ref{the:PtoLRBEStrans2} give us an operational correspondence between a $\pi$IK process and the generated event structure. Together with Theorem~\ref{the:key-ext-eq}, this gives us a correspondence between a $\pi$IH process and the event structure it generates by going via a $\pi$IK process.

As we want to ensure that all free and bound names in our process are distinct, we modify our syntax for replication, assigning each replication an infinite set, $\mathbf{x}$, of names to substitute into the place of bound names in each created copy of the process, so that $$!_{\mathbf{x}} P\equiv {!_{\mathbf{x}\setminus \{x_0,\dots, x_k\}} P}\vert {P\{\rfrac{x_0,\dots,x_k}{a_0,\dots,a_k}\}}\text{ if }\{x_0,\dots, x_k\}\subseteq \mathbf{x}\text{ and }\bn(P)=\{a_0,\dots,a_k\}$$

Before proceeding to the semantics we also define the standard bound names of a process $P$, $\mathsf{sbn}(P)$, meaning the names that would be bound in $P$ if every action was reversed, in Definition~\ref{def:sbn}.

\begin{definition}\label{def:sbn}
The standard bound names of a process $P$, $\mathsf{sbn}(P)$, are defined as:\vspace{5pt}

\begin{tabular}{ll}
$\mathsf{sbn}(a(x).P')=\{x\}\cup \mathsf{sbn}(P')$ & $\mathsf{sbn}(a(x)[m].P')=\{x\}\cup \mathsf{sbn}(P')$ \\ $\mathsf{sbn}(\overline{a}(x).P')=\{x\}\cup \mathsf{sbn}(P')$ &
$\mathsf{sbn}(\overline{a}(x)[m].P')=\{x\}\cup \mathsf{sbn}(P')$ \\ $\mathsf{sbn}(P_0\vert P_1)=\mathsf{sbn}(P_0) \cup \mathsf{sbn}(P_1)$ \hspace{10pt} & $\mathsf{sbn}(P_0 + P_1)=\mathsf{sbn}(P_0) \cup \mathsf{sbn}(P_1)$ \\
$\mathsf{sbn}( \nu x) P' =\{x\} \cup \mathsf{sbn}(P')$ & $\mathsf{sbn}(!_{\mathbf{x}}P)= \mathbf{x}$ \\
\end{tabular}
\end{definition}

	We can now define the event structure semantics in Table~\ref{tab:pik-den-sem}. We do this using rules of the form $\lrBraces{P}_{(\M{N},l)}=\lrangles{\M{E},\I,k}$ where $l$ is the level of unfolding of replication, $\M{E}$ is an LRBES, $\I$ is the initial configuration, $\M{N}\supseteq n(P)$ is a set of names, which any input in the process could receive, and $k:\I\rightarrow\mathcal{K}$ is a function assigning communication keys to the past actions, which we use in parallel composition to determine which synchronisations of past actions to put in $\I$. We define $\lrBraces{P}_{\M{N}}=\sup_{l\in \mathbb{N}} \lrBraces{P}_{(\M{N},l)}$
	
	{The denotational semantics in Table~\ref{tab:pik-den-sem} make use of of the LRBES operators defined in Section~\ref{sec:BES}. The choice and output cases are straightforward uses of the choice and causal prefix operators. The input creates a case for prefixing an input of each name in $\M{N}$ and a choice between the cases. We have two cases for restriction, one for restriction originating from a past communication and another for restriction originating from the original process. If the restriction does not originate from the original process, then we ignore it, otherwise we remove events which would use the restricted channel and their causes. The parallel composition uses the parallel operator, but additionally needs to consider link causation caused by the early semantics.} Each event labelled with an input of a name in standard bound names gets a bundle consisting of the event labelled with the output on that name. And each output event is prevented from reversing by the input names receiving that name. {This way, inputs on extruded names are caused by the output that made the name free. Replication substitutes the names and counts down the level of replication.}
	
	Note that the only difference between a future and a past action is that the event corresponding to a past action is put in the initial state and given a communication key.
	
\begin{table}{\small\begin{tabular}{ll}
 \hspace{-0cm}$\lrBraces{0}_{(\M{N},l)}=$ & $\lrangles{(\emptyset,\emptyset,\emptyset,\emptyset,\emptyset, \emptyset, \emptyset),\emptyset,\emptyset}$ \vspace{5pt}\\
 \hspace{-0cm}$\lrBraces{P_0+P_1}_{(\M{N},l)}=$ & $\lrangles{ \M{E}_0+\M{E}_1,\{0\}\times\I_0\cup \{1\}\times \I_1, k((i,e))=k_i(e)}$ where \\
  \hspace{-0cm} & $\lrBraces{P_i}=\lrangles{\M{E}_i,\I_i,k_i}$ for $i\in \{0,1\}$ \vspace{5pt}\\ 
 \hspace{-0cm}$\lrBraces{\outp{a}{n}.P}_{(\M{N},l)}=$ & $\lrangles{\outp{a}{n}(e).\M{E}_{P},\I_{P},k_{P}}$ for some fresh $e\notin E$ where \\
\hspace{-0cm} & $\lBrace P \rBrace_{(\M{N},l)}= \lrangles{\M{E}_{P},\I_{P},k_{P}}$ \vspace{5pt}\\
\hspace{-0cm}$\lrBraces{\inp{a}{x}.P}_{(\M{N},l)}=$ & $\lrangles{\sum\limits_{n\in (\M{N}\setminus \mathsf{sbn}(P))} \inp{a}{n}(e).\M{E}_{P_{n}},\bigcup\limits_{n\in (\M{N}\setminus \mathsf{sbn}(P))}\{n\}\times\I_{P_{n}},(n,e)\mapsto k_{P_{n}}(e)}$ \\ & for some fresh $e_n\notin E_n$ where \\
\hspace{-0cm} & $\lBrace P[x:= n] \rBrace_{(\M{N},l)}= \lrangles{\M{E}_{P_{n}},\I_{P_{n}},k_{P_{n}}}$ \vspace{5pt}\\
\hspace{-0cm}$\lrBraces{\outp{a}{n}[m].P}_{(\M{N},l)}=$\hspace{-0cm} & $\lrangles{\outp{a}{n}(e).\M{E}_{P},\I_{P}\cup \{e\},k_{P}[e\mapsto m]}$ for some fresh $e\notin E$ where \\
\hspace{-0cm} & $\lBrace P \rBrace_{(\M{N},l)}= \lrangles{\M{E}_{P},\I_{P},k_{P}}$ \vspace{5pt}\\
\hspace{-0cm}$\lrBraces{\inp{a}{b}[m].P}_{(\M{N},l)}=$\hspace{-0cm} & $\lrangles{\sum\limits_{n\in (\M{N}\setminus \mathsf{sbn}(P))} \inp{a}{n}(e_n).\M{E}_{P_{n}},(\bigcup\limits_{n\in (\M{N}\setminus \mathsf{sbn}(P))}\{n\}\times\I_{P_{n}})\cup \{(b,e_b)\},k}$ \\ & for some fresh $e_n\notin E_n$ where \\
\hspace{-0cm} & $\lBrace P[b_{[m]}:= n] \rBrace_{(\M{N},l)}= \lrangles{\M{E}_{P_{n}},\I_{P_{n}},k_{P_{n}}}$ \\
& $k((n,e))=\begin{dcases*}
m & if $e=e_b$ and $n=b$\\
k_{P_{n}}(e) & otherwise \\
\end{dcases*}$ \vspace{5pt}\\

\hspace{-0cm} $\lrBraces{(\nu a)P}_{(\M{N},l)}=$ & $\lrangles{\M{E}\upharpoonright E_{\alpha},\I\cap E_{\alpha}, k\upharpoonright E_{\alpha})}$ where: \\
\hspace{-0cm} & $\lBrace P \rBrace_{(\M{N},l)}= \lrangles{\M{E},\I,k}$ \\
\hspace{-0cm} & $E_{\alpha}=\rho(\{\alpha\mid a\in n(\alpha)\}$ \\
\hspace{-0cm} & if whenever there exist past actions $\inp{b}{a}[m]$ and $\outp{b}{a}[m]$ in $P$ then \\ & they are guarded by a restriction $(\nu a)$ in $P$\vspace{5pt}\\
\hspace{-0cm} $\lrBraces{(\nu a)P}_{(\M{N},l)}=$ & $\lrangles{\M{E},\I, k}$ where: \\
\hspace{-0cm} & $\lBrace P \rBrace_{(\M{N},l)}= \lrangles{\M{E},\I,k}$ \\
\hspace{-0cm} & if there exist past actions $\inp{b}{a}[m]$ and $\outp{b}{a}[m]$ in $P$ which\\ & are not guarded by a restriction $(\nu a)$ in $P$\vspace{5pt}\\
\hspace{-0cm} $\lrBraces{P_0\vert P_1}_{(\M{N},l)}=$ & $\lrangles{\LRBES{}\upharpoonright \{e\mid \lambda(e)\neq 0\},\I,k}$  where\\
\hspace{-0cm} & for $i\in \{0,1\}$, $\lBrace P_i \rBrace_l= \lrangles{\M{E}_{i} ,\I_i,k_i}$ \\
 \hspace{-0cm} & $(E_0,F_0,\mapsto_0,\cf_0,\rhd_0)\times (E_0,F_0,\mapsto_0,\cf_0,\rhd_0)=(E,F,\mapsto',\cf,\rhd')$ \\
 & $\I=\{(e_0,*)\vert e_0\in \I_0\text{ and } \nexists e_1\in \I_1.k_1(e_1)=k_0(e_0)\} \cup$ \\ & $ \{(*,e_1)\vert e_1\in \I_1\text{ and } \nexists e_0\in \I_0.k_1(e_1)=k_0(e_0)\} \cup$ \\ & $\{(e_0,e_1)\vert e_0\in \I_0\text{ and } e_1\in \I_1 \text{ and } k_1(e_1)=k_0(e_0)\}$ \\
 & $X\mapsto e$ if $X\mapsto' e$ or there exists $x\in \mathsf{no}(\lambda(e))$ such that  \\ &$X=\{e'\mid\exists a. \lambda(e')=\outp{a}{x} \}$ and $x\in \mathsf{sbn}(P)$\\
 & $e\rhd \underline{e'}$ if either $e\rhd' \underline{e'}$ or there exists $x\in \mathsf{no}(\lambda(e))$ and $a$ such that $\lambda(e')=\overline{a}(x)$ \\
 \hspace{-0cm} & $k(e)=\begin{dcases*}
 k_0(e_0) & if $e=(e_0,*)$\\
 k_1(e_1) & if $e=(*,e_1)$\\
 k_0(e_0) & if $e=(e_0,e_1)$\\
 \end{dcases*}$ \vspace{5pt}\\
\hspace{-0cm} $\lrBraces{!_{\mathbf{x}}P}_{(\M{N},0)}=$ & $\lrangles{(\emptyset,\emptyset,\emptyset, \emptyset, \emptyset, \emptyset, \emptyset),\emptyset,\emptyset}$ \vspace{5pt}\\
\hspace{-0cm} $\lrBraces{!_{\mathbf{x}}P}_{(\M{N},l)}=$ & $\lrBraces{!_{\mathbf{x}\setminus \{x_0,\dots, x_k\}} P\vert P\{\rfrac{x_0,\dots,x_k}{a_0,\dots,a_k}\}}_{(\M{N},l-1)}$ if $\{x_0,\dots, x_k\}\subseteq \mathbf{x}$ \\ & and $\bn(P)=\{a_0,\dots,a_k\}$\vspace{5pt}\\
 \end{tabular}}
 \caption{Denotational event structure semantics of $\pi$IK}\label{tab:pik-den-sem}
 \end{table}

\begin{example}
 Consider the process $\inp{a}{b}[n]\mid \outp{a}{b}[n]$. Our event structure semantics generate a LRBES $\lrBraces{\inp{a}{x}[n]\mid \outp{a}{b[n]}}_{\{a,b,x\}}=\lrangles{\LRBES{},\I,k}$ where:
\[ \begin{array}{lrl}
 E=F = \{\inp{a}{b},\inp{a}{a},\inp{a}{x},\outp{a}{b},\tau\} &\lambda(e)=& e\\
  \{\outp{a}{b}\}\mapsto \inp{a}{b} & \mathsf{Act}=& \{\inp{a}{b},\inp{a}{a},\inp{a}{x},\outp{a}{b},\tau\}\\
  \inp{a}{b}\cf \inp{a}{a},~\inp{a}{b}\cf\inp{a}{x},~\inp{a}{a}\cf\inp{a}{x}, & \I=&\{\tau\} \\
  \inp{a}{b}\cf\tau,~\inp{a}{a}\cf\tau,~\inp{a}{x}\cf\tau,~\outp{a}{b}\cf\tau \;\;\;\;\;\; & k(\tau)=&n\\
  \inp{a}{b} \rhd \outp{a}{b} \\
 \end{array}\] 
 From this we see that (1) receiving $b$ is causally dependent on sending $b$, (2) all the possible inputs on $a$ are in conflict with one another, (3) the synchronisation between the input and the output is in conflict with either happening on their own, and (4) since the two past actions have the same key, the initial state contains their synchronisation.
 \end{example}
 
	We show in Theorems~\ref{the:PtoLRBEStrans} and~\ref{the:PtoLRBEStrans2} that given a process $P$ with a conflict-free initial state, including any reachable process,  performing a transition $P\xrightarrow{\mu[m]} P'$ does not affect the event structure, as $\lrBraces{P}_{\M{N}}$ and $\lrBraces{P'}_{\M{N}}$ are isomorphic. It also means we have an event $e$ labelled $\mu$ such that $e$ is available in $P$'s initial state, and $P'$'s initial state is $P$'s initial state with $e$ added. A similar event can be removed to correspond to a reverse action.
	
\begin{theorem}\label{the:PtoLRBEStrans}
Let $P$ be a forwards reachable process wherein all bound and free names are different and let $\M{N}\supseteq n(P)$ be a set of names. If (1) $\lrBraces{P}_{\M{N}}=\lrangles{\M{E},\I,k}$ where $\M{E}=\LRBES{}$, and $\I$ is conflict-free, and (2) there exists a transition $P\xrightarrow{\mu[m]} P'$ such that $\lrBraces{P'}_{\M{N}}=\lrangles{\M{E}',\I',k'}$, then there exists an isomorphism $f:\M{E}\rightarrow \M{E}'$ and a transition in $C_{br}(\M{E})$, $\I\xrightarrow{\{e\}} X$, such that $\lambda(e)=\mu$, $f\circ k'=k[e\mapsto m]$, and $f(X)=\I'$.
\end{theorem}
\begin{theorem}\label{the:PtoLRBEStrans2}
Let $P$ be a forwards reachable process wherein all bound and free names are different and let $\M{N}\supseteq n(P)$ be a set of names. If (1) $\lrBraces{P}_{\M{N}}=\lrangles{\M{E},\I,k}$ where $\M{E}=\LRBES{}$, and (2) there exists a transition $\I\xrightarrow{\{e\}} X$ in $C_{br}(\M{E})$, then there exists a transition $P\xrightarrow{\mu[m]} P'$ such that $\lrBraces{P'}_{\M{N}}=\lrangles{\M{E}',\I',k'}$ and an isomorphism $f:\M{E}\rightarrow \M{E}'$ such that $\lambda(e)=\mu$, $f\circ k'=k[e\mapsto m]$, and $f(X)=\I'$.
\end{theorem}

By Theorems~\ref{the:key-ext-eq},~\ref{the:PtoLRBEStrans}, and~\ref{the:PtoLRBEStrans2} we can combine the event structure semantics of $\pi$IK and mapping $E$ (Definition~\ref{def:E}) and get an operational correspondence between $\mathbf{H}\vdash\! P$ and the event structure $\lrBraces{E(\mathsf{lcopy}(\mathbf{H})\vdash\! P,P)}_{n(E(\mathsf{lcopy}(\mathbf{H})\vdash\! P,P))}$.

\section{Conclusion and future work}
All existing reversible versions of the $\pi$-calculus use reduction semantics~\cite{lanese2010reversing,TIEZZI2015684} or late semantics~\cite{cristescu2013compositional,DBLP:journals/corr/abs-1808-08655}, despite the early semantics being used more widely than the late in the forward-only setting.
We have introduced $\pi$IH, the first reversible early $\pi$-calculus. It is a reversible form of the \emph{internal} $\pi$-calculus,
where names being sent in output actions are always bound.
As well as structural causation, as in CCS, the early form of the internal $\pi$-calculus also has a form of link causation {created by the semantics being early, which is not present in other reversible $\pi$-calculi}.
In $\pi$IH past actions are tracked by using extrusion histories adapted from~\cite{hildebrandt2017stable}, which move past actions and their locations into separate histories for dynamic reversibility.
We mediate the event structure semantics of $\pi$IH via
a statically reversible version of the internal $\pi$-calculus, $\pi$IK, which keeps the structure of the process intact but annotates past actions with keys,
similarly to $\pi$K~\cite{DBLP:journals/corr/abs-1808-08655} and CCSK~\cite{PU07}.
We showed that a process $\pi$IH with extrusion histories can be mapped to a $\pi$IK process with keys, creating an operational correspondence (Theorem~\ref{the:key-ext-eq}).


The event structure semantics of $\pi$IK, and by extension $\pi$IH, are defined inductively on the syntax of the process.
We use labelled reversible bundle event structures~\cite{EFG2018},
rather than prime event structures, to get a more compact representation where each action in the calculus has only one corresponding event.
While causation in the internal $\pi$-calculus is simpler that in the full $\pi$-calculus, our early semantics means that we still have to handle link causation, in the form of an input receiving a free name being caused by a previous output of that free name. We show an operational correspondence between $\pi$IK processes and their event structure representations in Theorems~\ref{the:PtoLRBEStrans} and~\ref{the:PtoLRBEStrans2}.
Cristescu \emph{et al.}~\cite{CristescuKV16} have used rigid families~\cite{CastellanHLW14}, related to event structures, to describe the semantics of R$\pi$~\cite{cristescu2013compositional}. However, unlike our denotational event structure semantics, their semantics require one to reverse every action in the process before applying the mapping to a rigid family, and then redo every reversed action in the rigid family. Our approach of using a static calculus as an intermediate step means we get the current state of the event structure immediately, and do not need to redo the past steps.

\paragraph{Future work:} We could expand the event structure semantics of $\pi$IK to $\pi$K. This would entail significantly more link causation, but would give us event structure semantics of a full $\pi$-calculus.
Another possibility is to expand $\pi$IH to get a full reversible early $\pi$-calculus. 

\paragraph{Acknowledgements:} We thank Thomas Hildebrandt and H{\aa}kon Normann for discussions on how to translate their work on $\pi$-calculus with extrusion histories to a reversible setting. 
We thank the anonymous reviewers of RC 2020 for their helpful comments.

This work was partially supported by an EPSRC DTP award; also by the following EPSRC projects: EP/K034413/1, EP/K011715/1, EP/L00058X/1, EP/N027833/1, EP/T006544/1, EP/N028201/1 and EP/T014709/1; and by EU COST Action IC1405 on Reversible Computation.

\bibliography{bib}
\bibliographystyle{splncs04}
\appendix
\section{Section~\ref{sec:Ext-sem}}\label{app:Ext-sem}

\begin{lemma}\label{lem:SingleExtrusion}
Let $P$ be a process. If there exists an extrusion history $\mathbf{H}$ such that $\mathbf{H}\vdash\! P\xsquigarrow{\mu}{u} \mathbf{H}'\vdash\! P'$ then there exists $L$ such that $\mathbf{H}=\mathbf{H}'+L$, and for any extrusion history $\mathbf{H}''$ not containing $L$, $\mathbf{H}''+ L\vdash\! P\xsquigarrow{\mu}{u} \mathbf{H}''\vdash\! P'$.
\end{lemma}
 \begin{proof}

[SCOPE$^{-1}$] and [PAR$^{-1}_i$] simply propagate the changes to extrusion histories, and [COM$^{-1}_i$], [IN$^{-1}$], and [OUT$^{-1}$] remove exactly one extrusion from the histories, which is the only one they depend on.
\end{proof}

 \begin{proof}[Proof of Proposition~\ref{prop:fwdtorevTrans}]
\begin{enumerate}
\item We prove this by induction in $\mathbf{H}\vdash\! P\,\xrightarrow[u]{\alpha} \mathbf{H}' \vdash\! Q$:
 
 \begin{description}
 \item[{[SCOPE]}] In this case $P=(\nu x) P'$ and $Q=(\nu x) Q'$, $x\notin n(\alpha)$, and by induction $\mathbf{H}'\vdash\! Q'\xsquigarrow{\alpha}{u} \mathbf{H}\vdash\! P'$. From rule [SCOPE$^{-1}$] we therefore get $\mathbf{H}'\vdash\! Q\xsquigarrow{\alpha}{u} \mathbf{H}\vdash\! P$.
 \item[{[PAR$_i$]}] In this case $P=P_0\vert P_1$ and $Q=Q_0\vert Q_1$, $P_{1-i}=Q_{1-i}$, in $\alpha=\overline{a}(n)$ then $n\notin \fn(P_{1-i})$, and by induction $([\check{i}]\overline{H}',[\check{i}]\underline{H}',[\check{i}]H')\vdash\! Q_i \xsquigarrow{\alpha}{u} ([\check{i}]\overline{H},[\check{i}]\underline{H},[\check{i}]H)\vdash\! P_i$, meaning according to rule [PAR$_i^{-1}$], $\mathbf{H}'\vdash\! Q\xsquigarrow{\alpha}{u} \mathbf{H}\vdash\! P$.
 \item[{[COM$_i$]}] In this case $P=P_0\vert P_1$ and $Q=Q_0\vert Q_1$, $n\notin \fn(P_j)$, $\overline{H}=\overline{H}'$, $\underline{H}=\underline{H}'$, $H'=H\cup0 \{(n,(0v_0,1v_1))\}$, and by induction and Lemma~\ref{lem:SingleExtrusion}, we have $([\check{i}]\overline{H}',[\check{i}]\underline{H}',[\check{i}]H')\vdash\! Q_i\xsquigarrow{\outp{a}{n}}{v_i} ([\check{i}]\overline{H},[\check{i}]\underline{H},[\check{i}]H)\vdash\! P_i$ and $([\check{j}]\overline{H}',[\check{j}]\underline{H}',[\check{j}]H')\vdash\! Q_j\xsquigarrow{\inp{a}{n}}{v_j} ([\check{j}]\overline{H},[\check{j}]\underline{H},[\check{j}]H)\vdash\! P_j$. This means according to [COM$^{-1}$], $\mathbf{H}'\vdash\! Q\xsquigarrow{\alpha}{u} \mathbf{H}\vdash\! P$.
\item[{[STR]}] In this case $Q\equiv Q'$,$\mathbf{H}'\vdash\! Q' \xsquigarrow{\alpha}{u} \mathbf{H}\vdash\! P'$, and $P'\equiv P$, and by rule [STR$^{-1}$], $\mathbf{H}'\vdash\! Q\xsquigarrow{\alpha}{u} \mathbf{H}\vdash\! P$.
 \item[{[OUT]}] In this Case $P=\sum\limits_{i\in I} \alpha_i.P_i$, $Q=P_j$ and $\alpha=\outp{a}{n}=\alpha_j$ for some $j\in I$, and by [OUT$^{-1}$], $\mathbf{H}'\vdash\! Q\xsquigarrow{\alpha}{u} \mathbf{H}\vdash\! P$.
 \item[{[IN]}] Similar to [OUT].
 \end{description}
\item Similar to previous.
\end{enumerate} 
\end{proof} 
 
  \begin{definition}[Location]
 Given a location $u$, its set of paths is defined as
 $$\loc(u)=\begin{dcases}
 \{l\} & \text{if } u=l[P][P'] \\
 \{ll_0,ll_1\} & \text{if } u=l\lrangles{l_0[P_0][P_0'],l_1[P_1][P_0]} \\
 \end{dcases}$$
 \end{definition}
  To get causal semantics of $\pi$IH, we add a set of causes to each transition, consisting of the previous extrusions, from the output history, which extruded the names of the action.
 \begin{definition}[Causal semantics]\label{def:CauseSem}
 The early causal semantics consist of transitions of the form ${\mathbf{H}\vdash\! P} \,\xrightarrow[u,D]{\alpha} {\mathbf{H}'\vdash\! P}$ where ${\mathbf{H}\vdash\! P} \,\xrightarrow[u]{\alpha} {\mathbf{H}'\vdash\! P}$ and
 \begin{enumerate}
 \item $(n,u) \in D\Rightarrow \exists a.\;(\outp{a}{n},u) \in\overline{H}$;
 \item if $(n,l),(n,l')\in D$ then $l=l'$;
 \item $\mathsf{dom}(D)=\mathsf{dom}(\overline{H})\cap \mathsf{no}(\alpha)$ where $\mathsf{no}(\alpha)$ is the set of non-output names in $\alpha$, defined by $\mathsf{no}(\outp{a}{b})=\{a\}\setminus \{b\}$, $\mathsf{no}(\inp{a}{b})=\{a,b\}$ and $\mathsf{no}(\tau)=\emptyset$.
 \end{enumerate}
 \end{definition}
 
 \begin{definition}[Independence]\label{def:Ind}
 Two locations, $u_0$ and $u_1$, are independent if for all $l_0\in \loc(u_0)$ and $l_1\in \loc(u_1)$, there exist $l$, $l_0'$, $l_1'$ such that either $l_0=l0l_0'$ and $l_1=l1l_1'$ or $l_0=l1l_0'$ and $l_1=l0l_1'$.
 
 Two transitions $\,\xarrowtail{\alpha_0}{u_0,D_0}$ and $\,\xarrowtail{\alpha_1}{u_1,D_1}$ are independent if $u_0$ and $u_1$ are independent, there does not exist $n$ such that $D_i(n)=u_{1-i}$.
 \end{definition}
 \begin{proposition}[Forward diamond \cite{danos2004reversible}]\label{prop:square}
If $\mathbf{H}\vdash\! P\,\xarrowtail{\alpha_0}{u_0,D_0} \mathbf{H}_0\vdash\! P_0$ and $\mathbf{H}\vdash\! P\,\xarrowtail{\alpha_1}{u_1,D_1} \mathbf{H}_1\vdash\! P_1$ are independent transitions then there exists $\mathbf{H}'\vdash\! P'$ such that $\mathbf{H}_0\vdash\! P_0\,\xarrowtail{\alpha_1}{u_1,D_1} \mathbf{H}'\vdash\! P'$ and $\mathbf{H}_1\vdash\! P_1\,\xarrowtail{\alpha_0}{u_0,D_0} \mathbf{H}' \vdash\! P'$.
\end{proposition}
\begin{proof}[Proof of Proposition~\ref{prop:square}]
This proof is similar to Theorem 14 of \cite{hildebrandt2017stable}. We have a path $l$ such that $u_i=l0u_i'$ and $u_{1-i}=l1u_{1-i}'$. If $Q \vert R$ is the parallel composition at location $l$, then $([\check{l0}]\overline{H},[\check{l0}]\underline{H},[\check{l0}]H)\vdash\! Q\,\xarrowtail{\alpha_i}{u_i'}$ and $([\check{l1}]\overline{H},[\check{l1}]\underline{H},[\check{l1}]H)\vdash\! Q\,\xarrowtail{\alpha_{1-i}}{u_{1-i}'}$ and there does not exist $n$ such that $D_i(n)=u_{1-i}$ or there does not exist $n$ such that $D_{1-i}(n)=u_{i}$, and by [PAR$_i$] and [PAR$_i^-1$], this means $\mathbf{H}_0\vdash\! P_0\,\xrightarrow[u_1]{\alpha_1} \mathbf{H}_0'\vdash\! P'$ and $\mathbf{H}_1\vdash\! P_1\,\xrightarrow[u_0]{\alpha_0} \mathbf{H}_1'\vdash\! P'$ and by Lemma 15 of \cite{hildebrandt2017stable}, $\mathbf{H}_0'=\mathbf{H}_1'$.

\end{proof}
%
%

\begin{definition}[Trace equivalence]
We define trace equivalence $\sim$ as the least equivalence relation closed under composition such that:
\[\begin{array}{rcl}
\mathsf{t};\underline{\mathsf{t}} &\sim& \varepsilon_\mathsf{t} \\
\mathbf{H}\vdash\! P\xarrowtail{\alpha_0}{u_0,D_0}\xarrowtail{\alpha_1}{u_1,D_1} \mathbf{H}'\vdash\! P' 
&\sim& \mathbf{H}\vdash\! P\xarrowtail{\alpha_1}{u_1,D_1}\xarrowtail{\alpha_0}{u_0,D_0}  \mathbf{H}'\vdash\! P'\\
&& \text{ if } \xarrowtail{\alpha_1}{u_1,D_1} \text{ and }\xarrowtail{\alpha_0}{u_0,D_0} \text{ are independent}\\
\end{array}\]
\end{definition}

 \begin{proposition}[Parabola]\label{the:para}
 Let $\mathsf{t}$ be a trace, then there exists a forward trace $\mathsf{t}_f$ and a backward trace $\mathsf{t}_b$ such that $\mathsf{t}\sim\mathsf{t}_b;\mathsf{t}_f$.
 \end{proposition}
\begin{proof}
We say that $\mathsf{t}= \mathbf{H}\vdash\! P\xarrowtail{\alpha_0}{u_0,D_0}\mathbf{H}_0\vdash\! P_0\xarrowtail{\alpha_1}{u_1,D_1}\dots \xarrowtail{\alpha_n}{u_n,D_n}\mathbf{H}'\vdash\! Q$ and $\mathsf{t}_b;\mathsf{t}_f=\mathbf{H}\vdash\! P\xsquigarrow{\alpha_0'}{u_0',D_0'}\dots \xsquigarrow{\alpha_k'}{u_k',D_k'}\mathbf{H}_k'\vdash\! P_k'\,\xrightarrow[u_{k+1}',D_{k+1}']{\alpha_{k+1}'} \dots \,\xrightarrow[u_{m}',D_m']{\alpha_{m}'}\mathbf{H}\vdash\! Q$.

We prove that they are equivalent by induction on the number of pairs $\,\xrightarrow[u_i,D_i]{\alpha_i}$, $\xsquigarrow{\alpha_{i+1}}{u_{i+1},D_{i+1}}$ and the length of the trace.

If no such pair exists, then $\mathsf{t}=\mathsf{t}_b;\mathsf{t}_f$, otherwise we find the first such pair $\,\xrightarrow[u_i,D_i]{\alpha_i}$, $\xsquigarrow{\alpha_{i+1}}{u_{i+1}}$.

If $u_i=u_{i+1}$ and $\alpha_i=\alpha_{i+1}$ then by Proposition~\ref{prop:fwdtorevTrans}, $\mathbf{H}_{i-1}\vdash\! P_{i-1}=\mathbf{H}_{i+1}\vdash\! P_{i+1}$, and we have a shorter trace $\mathbf{H}\vdash\! P\xarrowtail{\alpha_0}{u_0,D_0}\mathbf{H}_0\dots \mathbf{H}_{i-1}\vdash\! P_{i-1} \xarrowtail{\alpha_{i+2}}{u_{i+2},D_{i+2}}\dots \xarrowtail{\alpha_n}{u_n,D_n}\mathbf{H}'\vdash\! Q\sim \mathsf{t}$.

If $u_i\neq u_{i+1}$ or $\alpha_i\neq \alpha_{i+1}$ then if $u_i\npreceq u_{i+1}$ and $u_{i+1}\npreceq u_i$ then by Proposition~\ref{prop:square}, we have a trace $\mathbf{H}\vdash\! P\xarrowtail{\alpha_0}{u_0,D_0}\mathbf{H}_0\dots \mathbf{H}_{i-1}\vdash\! P_{i-1}\xarrowtail{\alpha_{i+1}}{u_{i+1},D_{i+1}}\xarrowtail{\alpha_i}{u_i,D_i} \mathbf{H}_{i+1}\vdash\! P_{i+1}\xarrowtail{\alpha_{i+2}}{u_{i+2},D_{i+2}}\dots \xarrowtail{\alpha_n}{u_n,D_n}\mathbf{H}'\vdash\! Q\sim\mathsf{t}$. If $u_i\preceq u_{i+1}$ then, since $\,\xrightarrow[u_i,D_i]{\alpha_i}$ is the most recent action in $\mathbf{H}_i$, $u_i=u_{i+1}$ and $\alpha_i\neq \alpha_{i+1}$. If $u_{i+1}\preceq u_{i}$ then, if $u_{i+1}=l[P_a][P_b]$, $P_b$ is not the subprocess located at location $l$ of $P_i$, meaning there cannot exist a transition $\mathbf{H}_i\vdash\! P_i \xsquigarrow{\alpha_{i+1}}{u_{i+1},D_{i+1}}$.
\end{proof}

\section{Section~\ref{sec:piik}}\label{app:piik}

\begin{lemma}\label{lem:SingleOutput}
Given a forwards reachable process $P$, if $P\xrightarrow{\outp{a}{x}[n]}$ then there cannot exist a past output action $\outp{b}{x}[m]$ anywhere in $P$.
\end{lemma}
 \begin{proof}
This would require $\outp{b}{x}[m]$ to either prefix, be in parallel with, or be an alternative choice to $\outp{a}{x}$ in $P$. The first two cases are impossible due to the $\text{if }\mu=\outp{a}{x} \text{ then } x\notin n(\alpha)$ and $\text{if }\mu=\outp{a}{x} \text{ then } x\notin \fn(P_1)$ requirement in the rules for propagating $\outp{a}{x}[n]$ past past actions and parallel composition, and the last case is prevented by requiring alternative paths to be standard if we want to propagate an action past the choice.
\end{proof}

\begin{proof}[Proof of Proposition~\ref{the:FwdToRev}]
\begin{enumerate}
\item We perform induction on $P\xrightarrow{\mu[n]} Q$:
\begin{enumerate}
\item Suppose $P=\inp{a}{x}.P'$, $\mu=\inp{a}{b}$, $\mathsf{std}(P')$, $Q=\inp{a}{b}[n].Q'$, and $Q'=Q[x:= b_{[n]}]$. Then, since $x\notin n(Q')$, $Q\xrsquigarrow{\inp{a}{b}} P$.
\item Suppose $P=\outp{a}{x}.P'$, $\mu=\outp{a}{x}$, $\mathsf{std}(P')$, $Q=\outp{a}{x}[n].P'$. Then clearly $Q\xrsquigarrow{\outp{a}{x}} P$.
\item Suppose $P=\alpha[m].P'$, $P'\xrightarrow{\mu[n]} Q'$, $Q=\alpha[m].Q'$, $n\neq m$, and if $\mu=\outp{a}{x}$ then $x\notin n(\alpha)$. Then by induction $Q'\xsquigarrow{\mu[n]} P'$, and clearly $Q\xrsquigarrow{\mu[n]} P$.
\item Suppose $P=P_0\vert P_1$, $P_0\xrightarrow{\mu[n]} Q_0$, \textsf{fsh}$[n](P_1)$, $Q=Q_0\vert P_1$, and if $\mu=\outp{a}{x}$ then $x\notin \fn(P_1)$. Then by induction, $Q_0\xrsquigarrow{\mu[n]}$, and obviously $Q\xrsquigarrow{\mu[n]} P$.
\item Suppose $P=P_0\vert P_1$, $P_0\xrightarrow{\inp{a}{x}[n]} Q_0$, $P_1\xrightarrow{\outp{a}{x}[n]} Q_1$, $\mu=\tau$, and $Q=(\nu x)(Q_0\vert Q_1)$. Then by induction $Q_0\xrsquigarrow{\inp{a}{x}} P_0$ and $Q_1\xrsquigarrow{\outp{a}{x}} P_1$, meaning clearly $Q\xrsquigarrow{\mu[n]} P$.
\item Suppose $P=P_0+P_1$, $P_0\xrightarrow{\mu[n]} Q_0$, $\textsf{std}(P_1)$, and $Q=Q_0+P_1$. Then by induction $Q_0\xrsquigarrow{\mu[n]} P_0$, meaning $Q\xrsquigarrow{\mu[n]} P$.
\item Suppose $P=(\nu x) P'$, $P'\xrightarrow{\mu[n]} Q'$, $x\notin n(\mu)$, and $Q=(\nu x) Q'$. Then by induction $Q'\xrsquigarrow{\mu[n]} P'$, and we get $Q\xrsquigarrow{\mu[n]} P$.
\item Suppose $P\equiv P'$, $P'\xrightarrow{\mu[n]} Q'$, and $Q\equiv Q'$. Then by induction $Q'\xrsquigarrow{\mu[n]} P'$, and therefore $Q\xrsquigarrow{\mu[n]} P$.
\end{enumerate}
\item We prove this by induction on $P\xrsquigarrow{\mu[n]} Q$:
\begin{enumerate}
\item Suppose $P=\inp{a}{b}[n].P'$, $\mu=\inp{a}{b}$, $\textsf{std}(P')$, $x\notin n(P')$, $Q'=P'[b_{[n]}:=x]$, and $Q=\inp{a}{x}.Q'$. Then clearly $Q\xrightarrow{\mu[n]} P$.
\item Suppose $P=\outp{a}{x}[n].P'$, $\mu=\outp{a}{x}$, $\mathsf{std}(P')$, $Q=\outp{a}{x}.P'$. Then clearly $Q\xrightarrow{\outp{a}{x}} P$.
\item Suppose $P=\alpha[m].P'$, $P'\xrsquigarrow{\mu[n]} Q'$, $m\neq n$, and $Q=\alpha[n].Q'$. Then by induction, $Q'\xrightarrow{\mu[n]} P'$, and since $P$ is forwards reachable, if $\mu=\outp{a}{x}$ then $x\notin n(\alpha)$. This means $Q\xrightarrow{\mu[n]} P$.
\item Suppose $P=P_0\vert P_1$, $P_0\xrsquigarrow{\mu[n]} Q_0$, \textsf{fsh}$[n](P_1)$, $Q=Q_0\vert P_1$, and if $\mu=\outp{a}(x)$ then $x\notin \fn(P_1)$. Then by induction $Q_0\xrightarrow{\mu[n]}P_0$, and clearly $Q\xrightarrow{\mu[n]} P$.
\item Suppose $P=(\nu x)(P_0\vert P_1)$, $\mu=\tau$, $P_0\xrsquigarrow{\inp{a}{x}[n]} Q_0$, $P_1\xrsquigarrow{\outp{a}{x}[n]} Q_1$, and $Q=Q_0\vert Q_1$. Then by induction $Q_0\xrightarrow{\inp{a}{x}} P_0$ and $Q_1\xrightarrow{\outp{a}{x}} P_1$, meaning clearly $Q\xrightarrow{\mu[n]} P$.
\item Suppose $P=P_0+P_1$, $P_0\xrsquigarrow{\mu[n]} Q_0$, $\textsf{std}(P_1)$, and $Q=Q_0+P_1$. Then by induction $Q_0\xrightarrow{\mu[n]} P_0$, meaning $Q\xrightarrow{\mu[n]} P$.
\item Suppose $P=(\nu x) P'$, $P'\xrsquigarrow{\mu[n]} Q'$, $x\notin n(\mu)$, and $Q=(\nu x) Q'$. Then by induction $Q'\xrightarrow{\mu[n]} P'$, and we get $Q\xrightarrow{\mu[n]} P$.
\item Suppose $P\equiv P'$, $P'\xrsquigarrow{\mu[n]} Q'$, and $Q\equiv Q'$. Then by induction $Q'\xrightarrow{\mu[n]} P'$, and therefore $Q\xrightarrow{\mu[n]} P$.
\end{enumerate}
\end{enumerate}
\end{proof}

\begin{proposition}[Reverse diamond]\label{prop:RevDiamond}
Given forwards reachable processes $P$, $Q$, and $R$, if $P\xrsquigarrow{\mu[m]} Q$ and $P\xrsquigarrow{\mu'[n]} R$ and $m\neq n$, then there exists a process $S$ such that $Q\xrsquigarrow{\mu'[n]} S$ and $R\xrsquigarrow{\mu[m]} S$.
\end{proposition}
\begin{proof}[Proof of Proposition~\ref{prop:RevDiamond}]
We use structural induction on $P$ to prove both these at once:
\begin{enumerate}
\item Suppose $P=0$ or $P=\alpha.P'$. Then $P$ cannot do any backwards transitions.
\item Suppose $P=\alpha[o].P'$. Then either $\textsf{std}(P')$ and $n=m=o$, or $Q=a(b)[o].Q'$, $R=a(b)[o].R'$, $P'\xrsquigarrow{\mu[m]} Q'$, and $P'\xrsquigarrow{\mu'[n]} R'$, meaning by induction there exists $S'$ such that $Q'\xrsquigarrow{\mu'[n]} S'$ and $R' \xrsquigarrow{\mu[m]} S'$. We say that $S=\alpha[n].S'$, and the theorem holds.
\item Suppose $P=P_0+P_1$, then either $\textsf{std}(P_0)$, $P_1\xrsquigarrow{\mu[m]} Q_1$, $P\xrsquigarrow{\mu'[n]} R_1$, $Q=P_0+Q_1$, and $R=P_0+R_1$, or $\textsf{std}(P_1)$, $P_0\xrsquigarrow{\mu[m]} Q_0$, $P\xrsquigarrow{\mu'[n]} R_0$, $Q=Q_0+P_1$, and $R=R_0+P_1$. In the first case, by induction there exists an $S_1$ such that $Q_1\xrsquigarrow{\mu'[n]} S_1$ and $R_1\xrsquigarrow{\mu[m]} S_1$, and we define $S=P_0+S_1$, and theorem holds. The second case is similar.
\item Suppose $P=(\nu x) P'$. Then either (1) $P'\xrsquigarrow{\mu[m]} Q'$ and $x\notin n(\mu)$ and $Q=(\nu x) Q'$ or (2) $P'=P_0\vert P_1$, $P_i \xrsquigarrow{\inp{a}{x}[m]} Q_i$, $P_{1-i} \xrsquigarrow{\outp{a}{x}[m]} Q_{i-1}$, $\mu=\tau$, and $Q= Q_0\vert Q_1$, and either (a) $P'\xrsquigarrow{\mu'[n]} R'$ and $x\notin n(\mu')$ and $R=(\nu x) R'$ or (b) $P'=P_0\vert P_1$, $P_i \xrsquigarrow{\inp{a}{x}[n]} R_i$, $P_{1-i} \xrsquigarrow{\outp{a}{x}[n]} R_{i-1}$, $\mu'=\tau$, and $R= R_0\vert R_1$. 

In case 1a, by induction there exists $S'$ such that $Q'\xrsquigarrow{\mu'[n]} S'$ and $R' \xrsquigarrow{\mu[m]} S'$, and we define $S=(\nu x)S'$, and the theorem holds.

In case 1b, there exists $P_j$ such that $P_j\xrsquigarrow{\mu[m]} Q_j$, and $\textsf{fsh}[m](P_{1-j})$, and if $\mu=\outp{a}{x}$ then $x\notin \fn(P_1)$. If $j=i$ then by induction there exists an $S_i$ such that $Q_j\xrsquigarrow{\mu'[n]} S_i$ and $R_i \xrsquigarrow{\inp{a}{x}[m]} S_i$, and we define $S=S_i\vert R_{1-i}$, and the theorem holds.
If $I=1-j$, the argument is similar.

Case 2a is similar to case 1b.

Case 2b cannot occur because we cannot have more than one past action outputting the same name according to Lemma~\ref{lem:SingleOutput}.
\item Suppose $P=P_0\vert P_1$. Then there exists an $i$ such that either $P_i\xrsquigarrow{\mu[m]} Q_i$ and $P_i\xrsquigarrow{\mu'[n]} R_i$ and $Q=Q_i\vert P_{1-i}$ and $R=R_i\vert P_{1-i}$, or $P_i\xrsquigarrow{\mu[m]} Q_i$ and $P_{1-i}\xrsquigarrow{\mu'[n]} R_{1-i}$ and $Q=Q_i\vert P_{1-i}$ and $R=P_i\vert R_{1-i}$. 

In the first case, there exists $S_i$ such that $Q_i\xrsquigarrow{\mu'[n]} S_i$ and $R_i \xrsquigarrow{\mu[m]} S_i$, and we define $S=S_i\vert P_{1-i}$ and the theorem holds.

If the second case we define $S=Q_i\vert R_{1-i}$, and the theorem holds.
\end{enumerate}
\end{proof}
\begin{proposition}\label{prop:UniqueKey}
Given forwards reachable processes $P$, $Q$, and $R$, if $P\xrsquigarrow{\mu[m]} Q$ and $P\xrsquigarrow{\mu'[m]} R$ then $\mu=\mu'$ and $R\equiv Q$.
\end{proposition}
\begin{proof}
We prove this by structural induction:
\begin{enumerate}
\item Suppose $P=0$ or $P=\alpha.P'$. Then $P$ cannot do any reverse transitions.
\item Suppose $P=\alpha[n].P'$. Then either $\textsf{std}(P')$, meaning $\mu=\mu'=\alpha$, $n=m$, and $Q\equiv R$, or $P'\xrsquigarrow{\mu[m]} Q'$, $P'\xrsquigarrow{\mu'[m]}R'$, $Q=\alpha[n].Q'$, and $R=\alpha[n].R'$, and the result follows from induction.
\item Suppose $P=P_0+P_1$. Then the result follows from induction.
\item Suppose $P=(\nu x) P'$. Then either (1) $P'\xrsquigarrow{\mu[m]} Q'$ and $x\notin n(\mu)$ and $Q=(\nu x) Q'$ or (2) $P'=P_0\vert P_1$, $P_i \xrsquigarrow{\inp{a}{x}[m]} Q_i$, $P_{1-i} \xrsquigarrow{\outp{a}{x}[m]} Q_{i-1}$, $\mu=\tau$, and $Q= Q_0\vert Q_1$, and either (a) $P'\xrsquigarrow{\mu'[m]} R'$ and $x\notin n(\mu')$ and $R=(\nu x) R'$ or (b) $P'=P_0\vert P_1$, $P_i \xrsquigarrow{\inp{a}{x}[m]} R_i$, $P_{1-i} \xrsquigarrow{\outp{a}{x}[m]} R_{i-1}$, $\mu'=\tau$, and $R= R_0\vert R_1$.

In case 1a the result follows from induction.

In case 1b $P_j$ such that $P_j\xrsquigarrow{\mu[m]} Q_j$, and $\textsf{fsh}[m](P_{1-j})$, contradicting $P_{1-j} \xrightarrow{\alpha[m]} R_{1-j}$. Meaning this case cannot occur.

Similar for case 2a.

Case 2b follows from induction.

\item Suppose $P=P_0\vert P_1$. Then there exists an $i$ such that either $P_i\xrsquigarrow{\mu[m]} Q_i$ and $P_i\xrsquigarrow{\mu'[m]} R_i$ and $Q=Q_i\vert P_{1-i}$ and $R=R_i\vert P_{1-i}$, or $P_i\xrsquigarrow{\mu[m]} Q_i$ and $P_{1-i}\xrsquigarrow{\mu'[m]} R_{1-i}$ and $Q=Q_i\vert P_{1-i}$ and $R=P_i\vert R_{1-i}$. 

In the first case the result follows from induction. In the second case $P_i\xrsquigarrow{\mu[m]} Q_i$ requires $\textsf{fsh}[m](P_{1-i})$, which contradicts $P_{1-i}\xrsquigarrow{\mu'[m]} R_{1-i}$, meaning this case cannot occur.
\end{enumerate}
\end{proof}
\begin{theorem}[Parabola]\label{the:paraK}
Given processes $P$ and $Q$, such that $P\rightarrowtail^* Q$, there exists a process $R$ such that $P\rightsquigarrow^*R\rightarrow^*Q$.
\end{theorem}
\begin{proof}[Proof of Theorem~\ref{the:paraK}]
We say that $P\xarrowtail{\mu_0[m_0]}{} P_0 \dots \xarrowtail{\mu_n[m_n]}{} P_n= Q$ and perform induction on the length of the trace, the number of pairs $\xrightarrow{\mu_i[m_i]}\xrsquigarrow{\mu_{i+1}[m_{i+1}]}$ in the trace, and the location of the first such pair.

If no such pair exists then $R$ must exist.

Otherwise, we say that $\xrightarrow{\mu_i[m_i]}\xrsquigarrow{\mu_{i+1}[m_{i+1}]}$ is the fist such pair in the trace. We have 2 cases, either $m_i=m_{i+1}$ or not. 

If $m_i=m_{i+1}$ then by Propositions~\ref{the:FwdToRev} and~\ref{prop:UniqueKey}, $P_{i-1}=P_{i+1}$, and we therefore have a trace $P\xarrowtail{\mu_0[m_0]}{} P_0 \dots \xarrowtail{\mu_{i-1}[m_{i-1}]}{}P_{i-1}\xarrowtail{\mu_{i+2}[m_{i+2}]}{}\dots \xarrowtail{\mu_n[m_n]}{} P_n= Q$.

If $m_i\neq m_{i+1}$ then by Proposition~\ref{prop:RevDiamond} we have a trace $P\xarrowtail{\mu_0[m_0]}{} P_0 \dots P_{i-1}\xrsquigarrow{\mu_{i+1}[m_{i+1}]}{}\xrightarrow{\mu_{i}[m_{i}]}P_{i+1}\dots \xarrowtail{\mu_n[m_n]}{} P_n= Q$
\end{proof}
\section{Section~\ref{sec:sem-corr}}
In Lemma~\ref{lem:S} we demonstrate, that $S$ does indeed annotate any name, which was substituted for $x_1$, with $n$.

We also define the root of a $\pi$IK process as removing all keys from the process.
\begin{definition}[Root]
We say that a $\pi$IK process, $P$, has a root, $\mathsf{rt}(P)$, defined as:\vspace{5pt}

\begin{tabular}{llll}
$\mathsf{rt}(0)=0$\;\; & $\mathsf{rt}(!P)={!\mathsf{rt}(P)}$\;\; & $\mathsf{rt}(P_0\vert P_1)=\mathsf{rt}(P_0)\vert \mathsf{rt}(P_1)$  \;\; &
$\mathsf{rt}(\alpha.P)=\alpha.\mathsf{rt}(P)$ \\ 
\multicolumn{2}{l}{$\mathsf{rt}(P_0+P_1)=\mathsf{rt}(P_0)+\mathsf{rt}(P_1)$ \;\;} &
$\mathsf{rt}(\alpha[m].P)=\alpha.\mathsf{rt}(P)$ & $\mathsf{rt}((\nu x) P)=(\nu x)\mathsf{rt}(P)$ \\
\end{tabular}
\end{definition}
\begin{lemma}\label{lem:S}
Given a standard $\pi$IK process $P$, a $\pi$IK process $P'$, a series of substitutions $[x_1:=a_1][x_2:=a_2]\dots [x_k:=a_k]$, such that $\mathsf{rt}(P')\equiv P[x_1:=a_1][x_2:=a_2]\dots [x_k:=a_k]$ using the definition of $\equiv$ from Section \ref{sec:Ext-sem}, and a key $[n]$, we get that $S(P',P,[n],x_1)=P''$ for some $P''$ such that
\[
\mathsf{rt}(P'')\equiv P[x_1:=a_{1_{[n]}}][x_2:=a_2]\dots [x_k:=a_k]\ .
\]
\end{lemma}
\begin{proof}[Proof of Lemma~\ref{lem:S}]
We prove this by structural induction on $P$:
\begin{itemize}
\item Assume $P=0$. Then $P'=P[x_1:=a_1][x_2:=a_2]\dots [x_k:=a_k]=0$ and $S(P[x_1:=a_1][x_2:=a_2]\dots [x_k:=a_k],P,[n],x_1)=0$.
\item Assume $P=\inp{b}{c}.Q$ Then either $P'=\inp{d}{e}.Q'$, or $P'=\inp{d}{e}[m].Q'$, for some $d,e,m$. We then get $4$ cases: either $b=x_1$, $c=x_1$, $b=c=x_1$, or $b\neq x_1$ and $c\neq x_1$.

Assume $b=x_1$ and $c\neq x_1$. Then $d=a_1$, and
\[
S(P',P,[n],x_1)=\inp{d{[n]}}{c}.S(Q',Q,[n],x_1)\ ,
\]
and the result follows from induction.

Assume $c=x_1$ and $b\neq x_1$. Then, since $c$ is bound, $P[x_1:=a_1]=P[x_1:=a_{1_[n]}]=P$, and the result follows.

Assume $b=c=x_1$. Then $d=a_1$ and $Q[x_1:=a_1][x_2:=a_2]\dots [x_k:=a_k]=Q[x_2:=a_2]\dots [x_k:=a_k]=Q'$, and the result follows.
\item Assume $P=\outp{b}{c}.Q$. This is similar to the previous case.
\item Assume $P=\sum\limits_{i\in I} P_i$. Then the result follows trivially from induction.
\item Assume $P=P_0\vert P_1$. Then either $P'=P_0'\vert P_1'$, or $P_0\equiv{!P_1}$ and $P'=P_0'$. 

If $P'=P_0'\vert P_1'$ then the result follows trivially from induction. 

If ${P_0}={!P_1}$ and $P'=P_0$, then $P''=S\left(!P_0',P_0,[n],x\right)\vert S\left(P_0',P_1,[n],x\right)$, and the result follows from induction.
\item Assume $P=(\nu b) Q$. Then $P'=(\nu c) Q'$ and either $b=x_1$ or $b\neq x_1$. 

If $b=x_1$, then $P[x_1:=a_1]=P[x_1:=a_{1_[n]}]=P$.

If $b\neq x_1$, then the result follows from induction.

\item Assume ${P}={!Q}$. Then either ${P'}={!Q'}$, or $P'=P_0'\vert P_1'$.

If ${P'}={!Q'}$, the result follows trivially from induction.

Otherwise the case is similar to the second case on parallel composition. 
\end{itemize}
\end{proof}

\begin{proof}[Proof of Theorem~\ref{the:key-ext-eq}]
We first show that if there exists a location $u$ such that $\mathbf{H}\vdash\! P\,\xrightarrow[u]{\mu} \mathbf{H}'\vdash\! P'$, then there exists a key $m$, such that $E((\{(a,u,u)\mid (a,u)\in \overline{H}\},\{(a,u,u)\mid (a,u)\in \underline{H}\},\{(a,u,u)\mid (a,u)\in {H}\})\vdash\! P,P)\xarrowtail{\mu[m]} E((\{(a,u,u)\mid (a,u)\in \overline{H'}\},\{(a,u,u)\mid (a,u)\in \underline{H'}\},\{(a,u,u)\mid (a,u)\in H'\})\vdash\! P',P')$ by induction in the size of $\overline{H}\cup \underline{H}\cup H$ and the structure of $P$:

Assume $\mathbf{H}=(\emptyset,\emptyset,\emptyset)$. Then $E(\mathbf{H}\vdash\! P,P)=P$.
\begin{itemize}
\item Assume $P=a(x).Q$. Then $\mu=a(b)$, $u=[P][Q[x:=b]]$, and $\mathbf{H}'\vdash\! P'=(\emptyset,\{(a(b),u)\},\emptyset)\vdash\! Q[x:=b]$. We then by Lemma~\ref{lem:S} get $E(\mathbf{H}'\vdash\! P',P')=a(x)\left[[P][Q[x:=b]]\right].Q[x:=b_{\left[[P][Q[x:=b]]\right]}]$, and the rest of the case follows naturally.

\item Assume $P=\outp{a}{x}.Q$. This case is similar to the previous.

\item Assume $P=P_0\vert P_1$. Then either $u=iu'$, or $u=\lrangles{0u_0,1u_1}$. 

If $u=0u'$, then $(\emptyset,\emptyset,\emptyset)\vdash\! P_0 \,\xrightarrow[u']{\mu} \mathbf{H}_0'\vdash\! P_0'$, $\mathbf{H}'\vdash\! P'=(0\overline{H'}_0,0\underline{H'}_0,0H'_0)\vdash\! P_0'\vert P_1$, and if $\mu=\outp{a}{b}$ then $b\notin \fn(P_1)$. By induction, $P_0\xrightarrow{\mu[m]} E(\mathbf{H}_0'\vdash\! P_0',P_0)$, and therefore $P_0\vert P_1 \xrightarrow{\mu[m]} E(\mathbf{H}_0'\vdash\! P_0',P_0)\vert P_1=E((0\overline{H'}_0,0\underline{H'}_0,0H'_0)\vdash\! P_0'\vert P_1, P_0'\vert P_1)$.

If $u=1u'$, the case is similar to $u=0u'$.

If $u=\lrangles{0u_0,1u_1}$, then $(\emptyset,\emptyset,\emptyset)\vdash\! P_i \,\xrightarrow[u_i]{\inp{a}{b}} \mathbf{H}_i'\vdash\! P_i'$ and $(\emptyset,\emptyset,\emptyset)\vdash\! P_{1-i} \,\xrightarrow[u']{\outp{a}{b}} \mathbf{H}_{1-i}\vdash\! P_{1-i}'$ for some $i\in \{0,1\}$ and $b\notin \fn(P_i)$ and $\mathbf{H}'\vdash\! P'=(\emptyset,\emptyset,\{(\inp{a}{b},\outp{a}{b},u)\})\vdash\! P_0'\vert P_1'$. By induction, $E((\emptyset,\emptyset,\emptyset)\vdash\! P_i,P_i)\xrightarrow{a(b)[m]} E(\mathbf{H}_i'\vdash\! P_i',P_i')$ and $E((\emptyset,\emptyset,\emptyset)\vdash\! P_{1-i},P_{1-i})\xrightarrow{\outp{a}{b}[m]} E(\mathbf{H}_{1-i}'\vdash\! P_{1-i}',P_{i-1})$. Therefore $$\begin{array}{r}P_0\vert P_1\xrightarrow{\tau[m]} E((\emptyset,\emptyset,\{(a(b),\outp{a}{b},\lrangles{0u_0,1u_1},m)\})\vdash\! (\nu b) (P_0'\vert P_1'),(\nu b) (P_0'\vert P_1'))\\ =(\nu b) E((\emptyset,\emptyset,\{(a(b),\outp{a}{b},\lrangles{0u_0,1u_1},m)\})\vdash\!  (P_0'\vert P_1'),(P_0'\vert P_1'))\end{array}$$ 

\item Assume $P=(\nu x) Q$. Then $(\emptyset,\emptyset,\emptyset)\vdash\! Q\,\xrightarrow[u]{\mu} \mathbf{H}'\vdash\! Q'$, $x\notin n(\mu)$, and $P'=(\nu x) Q'$. We then get by induction $Q\xrightarrow{\mu[m]} E(\mathbf{H}'\vdash\! Q',Q')$, and therefore $(\nu x)Q \xrightarrow{\mu[m]} (\nu x)E(\mathbf{H}'\vdash\! Q',Q')=E(\mathbf{H}'\vdash\! (\nu x) Q',(\nu x) Q')$.
\item Assume ${P}={!Q}$. Then $(\emptyset,\emptyset,\emptyset)\vdash\! !Q\vert Q\,\xrightarrow[u]{\mu} \mathbf{H}'\vdash\! P'$, and the rest follows from the parallel case. 
\end{itemize}
If for any $(\mu',u')\in \overline{H}\cup \underline{H}\cup H$, if there exists a location $u$ such that $\mathbf{H}-(\mu',u')\vdash\! P\,\xrightarrow[u]{\mu} \mathbf{H}''\vdash\! P'$, then there exists a key $m$, such that $E(\mathbf{H}-(\mu',u')\vdash\! P,P)\xrightarrow{\mu[m]} E(\mathbf{H}''\vdash\! P',P')$, then $E$ only adds past actions and unused choice branches to the process, both of which one can easily propagate the action past.

We then show that if there exists a key, $m$, such that $E((\{(a,u,u)\mid (a,u)\in \overline{H}\},\{(a,u,u)\mid (a,u)\in \underline{H}\},\{(a,u,u)\mid (a,u)\in {H}\})\vdash\! P,P)\xarrowtail{\mu[m]}{}P''$, then there exists a location, $u$, and a $\pi$IH process, $\mathbf{H}'\vdash\! P'$, such that $\mathbf{H}\vdash\! P\xarrowtail{\mu}{u} \mathbf{H}'\vdash\! P'$ and $P''\equiv E((\{(a,u,u)\mid (a,u)\in \overline{H'}\},\{(a,u,u)\mid (a,u)\in \underline{H'}\},\{(a,u,u)\mid (a,u)\in H'\})\vdash\! P',P')$ We again do this by induction on the number of extrusions in $\overline{H}\cup \underline{H}\cup H$, and the structure of $P$.

Assume $\overline{H}\cup \underline{H}\cup H=\emptyset$. Then $E((\{(a,u,u)\mid (a,u)\in \overline{H}\},\{(a,u,u)\mid (a,u)\in \underline{H}\},\{(a,u,u)\mid (a,u)\in {H}\})\vdash\! P,P)=P$. Since we are only proving operational correspondence up to structural congruence, we can discount any rules employing that.
\begin{itemize}
\item Assume $P=a(x).Q$. Then $\mu=a(b)$, we select $m=[[P][Q[x:=b]]]$, and $E((\{(a,u,u)\mid (a,u)\in \overline{H'}\},\{(a,u,u)\mid (a,u)\in \underline{H'}\},\{(a,u,u)\mid (a,u)\in H'\})\vdash\! P',P')=a(b)[[P][Q[x:=b]]].Q[x:=b_{[[P][Q[x:=b]]]}]$. We say that $u=[P][Q[x:=b]]$, $\mathbf{H}'\vdash\! P'=(\emptyset,\{(\inp{a}{b},[P][Q[x:=b]],[P][Q[x:=b]])\},\emptyset)\vdash\! Q[x:=b]$, and by Lemma~\ref{lem:S} the result follows.
\item Assume $P=\outp{a}{b}.Q$. This case is similar to the previous.
\item Assume $P=P_0\vert P_1$. Then either $P_0\xrightarrow{\mu[m]} P_0'$ and $E((\{(a,u,u)\mid (a,u)\in \overline{H'}\},\{(a,u,u)\mid (a,u)\in \underline{H'}\},\{(a,u,u)\mid (a,u)\in H'\})\vdash\! P',P')=P_0'\vert P_1$, $P_1\xrightarrow{\mu[m]} P_1'$ and $E((\{(a,u,u)\mid (a,u)\in \overline{H'}\},\{(a,u,u)\mid (a,u)\in \underline{H'}\},\{(a,u,u)\mid (a,u)\in H'\})\vdash\! P',P')=P_0\vert P_1'$, or $P_i\xrightarrow{a(b)[m]} P_i'$, $P_{1-i}\xrightarrow{\outp{a}{b}[m]} P_{1-i}'$, $\mu=\tau$, and $E((\{(a,u,u)\mid (a,u)\in \overline{H'}\},\{(a,u,u)\mid (a,u)\in \underline{H'}\},\{(a,u,u)\mid (a,u)\in H'\})\vdash\! P',P')=(\nu b) (P_0'\vert P_1')$. 

If $P_0\xrightarrow{\mu[m]} P_0'$ and $E((\{(a,u,u)\mid (a,u)\in \overline{H'}\},\{(a,u,u)\mid (a,u)\in \underline{H'}\},\{(a,u,u)\mid (a,u)\in H'\})\vdash\! P',P')=P_0'\vert P_1$, then by induction, there exists $u_0$ and $\mathbf{H}_0\vdash\! P_0''$ such that $E((\{(a,u,u)\mid (a,u)\in \overline{H'}_0\},\{(a,u,u)\mid (a,u)\in \underline{H'}_0\},\{(a,u,u)\mid (a,u)\in H'_0\})\vdash\! P_0'',P_0'')=P_0'$ and $(\emptyset,\emptyset,\emptyset)\vdash\! P_0\,\xrightarrow[u_0]{\mu} \mathbf{H}_0\vdash\! P_0''$. We therefore get $(\emptyset,\emptyset,\emptyset)\vdash\! P\,\xrightarrow[0u_0]{\mu} (0\overline{H_0},0\underline{H_0},0H_0)\vdash\! P_0'\vert P_1$.

If $P_1\xrightarrow{\mu[m]} P_1'$ and $E((\{(a,u,u)\mid (a,u)\in \overline{H'}\},\{(a,u,u)\mid (a,u)\in \underline{H'}\},\{(a,u,u)\mid (a,u)\in H'\})\vdash\! P',P')=P_0\vert P_1'$, then the case is similar to the previous.

If $P_i\xrightarrow{a(b)[m]} P_i'$, $P_{1-i}\xrightarrow{\outp{a}{b}[m]} P_{1-i}'$, $\mu=\tau$, and $E((\{(a,u,u)\mid (a,u)\in \overline{H'}\},\{(a,u,u)\mid (a,u)\in \underline{H'}\},\{(a,u,u)\mid (a,u)\in H'\})\vdash\! P',P')=(\nu b) (P_0'\vert P_1')$, then by induction we have $u_i$ and $\mathbf{H}_i\vdash\! P_i''$, $u_{1-i}$ and $\mathbf{H}_{1-i}\vdash\! P_{1-i}''$ such that $E((\{(a,u,u)\mid (a,u)\in \overline{H'}_0\},\{(a,u,u)\mid (a,u)\in \underline{H'}_0\},\{(a,u,u)\mid (a,u)\in H'_0\})\vdash\! P_0'',P_0'')=P_0'$ and $E((\{(a,u,u)\mid (a,u)\in \overline{H'}_1\},\{(a,u,u)\mid (a,u)\in \underline{H'}_1\},\{(a,u,u)\mid (a,u)\in H'_1\})\vdash\! P_1'',P_1'')=P_1'$ and $(\emptyset,\emptyset,\emptyset)\vdash\! P_i\,\xrightarrow[u_i]{a(b)} \mathbf{H}_i\vdash\! P_i''$ and $(\emptyset,\emptyset,\emptyset)\vdash\! P_{1-i}\,\xrightarrow[u_{1-i}]{\outp{a}{b}} \mathbf{H}_{1-i}\vdash\! P_{1-i}''$. We therefore say $m=\lrangles{0u_0,1u_1}$, and get $(\emptyset,\emptyset,\emptyset)\vdash\! P_0\vert P_1\,\xrightarrow[\lrangles{0u_0,1u_1}] (\emptyset,\emptyset,\{(a(b),\outp{a}{b},\lrangles{0u_0,1u_1},m)\})\vdash\! (\nu b) (P_0''\vert P_1'')$.
\item Assume $P=(\nu x) Q$. Then $x\notin n(\mu)$, $E((\{(a,u,u)\mid (a,u)\in \overline{H'}\},\{(a,u,u)\mid (a,u)\in \underline{H'}\},\{(a,u,u)\mid (a,u)\in H'\})\vdash\! P',P')= (\nu x) Q'$ and $Q\xrightarrow{\mu[m]} Q'$. We therefore get $P'=(\nu x) Q''$, and by induction $\mathbf{H}\vdash\! Q\,\xrightarrow[u]{\mu} \mathbf{H}'\vdash\! Q'$, and therefore $\mathbf{H}\vdash\! P\,\xrightarrow[u]{\mu} \mathbf{H}'\vdash\! P'$.
\item Assume ${P}={!Q}$. Then the transition must involve structural congruence, $!Q\vert Q\xrightarrow{\mu[m]} P'''$ for $P'''\equiv P''$, and the rest follows from the parallel case.
\end{itemize}

If for any $(\mu',u')\in \overline{H}\cup \underline{H}\cup H$, if there exists a key $m$, such that $E(\mathbf{H}-(\mu',u')\vdash\! P,P)\xrightarrow{\mu[m]} E(\mathbf{H}''\vdash\! P',P')$, then there exists a location $u$ such that $\mathbf{H}-(\mu',u')\vdash\! P\,\xrightarrow[u]{\mu} \mathbf{H}''\vdash\! P'$, then having more past extrusions does not stop $\mathbf{H}-(\mu',u')\vdash\! P$ from performing any forwards actions and having more past actions does not allow $E(\mathbf{H}-(\mu',u')\vdash\! P,P)$ to perform additional forward actions.

We then need to prove that if there exists a location $u$ such that $\mathbf{H}\vdash\! P\xsquigarrow{\mu}{u} \mathbf{H}'\vdash\! P'$, then there exists a key $m$, such that $E((\{(a,u,u)\mid (a,u)\in \overline{H}\},\{(a,u,u)\mid (a,u)\in \underline{H}\},\{(a,u,u)\mid (a,u)\in {H}\})\vdash\! P,P)\xrsquigarrow{\mu[m]} E((\{(a,u,u)\mid (a,u)\in \overline{H'}\},\{(a,u,u)\mid (a,u)\in \underline{H'}\},\{(a,u,u)\mid (a,u)\in H'\})\vdash\! P',P')$.

This follows naturally from the above properties, and Propositions~\ref{the:FwdToRev} and~\ref{prop:fwdtorevTrans}.

We finally need to prove that if there exists a key, $m$, such that $E((\{(a,u,u)\mid (a,u)\in \overline{H}\},\{(a,u,u)\mid (a,u)\in \underline{H}\},\{(a,u,u)\mid (a,u)\in {H}\})\vdash\! P,P)\xrsquigarrow{\mu[m]}{}P''$, then there exists a location, $u$, and a $\pi$IH process, $\mathbf{H}'\vdash\! P'$, such that $\mathbf{H}\vdash\! P\xsquigarrow{\mu}{u} \mathbf{H}'\vdash\! P'$ and $P''\equiv E((\{(a,u,u)\mid (a,u)\in \overline{H'}\},\{(a,u,u)\mid (a,u)\in \underline{H'}\},\{(a,u,u)\mid (a,u)\in H'\})\vdash\! P',P')$.

As we have proven the above properties, and Propositions~\ref{the:FwdToRev}, and~\ref{prop:fwdtorevTrans}, we only need to prove that there exists a $\pi$IH process $\mathbf{H}'\vdash\! P'$, such that $P''\equiv E((\{(a,u,u)\mid (a,u)\in \overline{H'}\},\{(a,u,u)\mid (a,u)\in \underline{H'}\},\{(a,u,u)\mid (a,u)\in H'\})\vdash\! P',P')$. Since none of the transition rules - forward or reverse - in $\pi$IK can create unguarded choice from guarded choice, and $E$ only generates $\pi$I-calculus processes with guarded choice, we know $P''$ has guarded choice. 

If $P''$ is a standard process, then $\mathbf{H}'\vdash\! P'=(\emptyset,\emptyset,\emptyset)\vdash\! P''$. Otherwise, by Theorems~\ref{the:paraK} and~\ref{the:para}, $P''$ must be forwards reachable from a standard process $P'''$ such that $P'''\equiv E((\emptyset,\emptyset,\emptyset)\vdash\! P''',P''')$, and by the above properties, $\mathbf{H}'\vdash\! P'$ exists.

%
%

\end{proof}

\section{Section~\ref{sec:Den-Ev-Sem}}

\begin{proposition}[Structural Congruence]\label{prop:RollStrCon}
Given processes $P$ and $P'$ and a set of names $\M{N}\supseteq n(P)\cup n(P')$, if $P\equiv P'$, $\lrBraces{P}_{\M{N}}=\lrangles{\M{E},\I,k}$, and $\lrBraces{P'}_{\M{N}}=\lrangles{\M{E}',\I',k'}$, then there exists an isomorphism $f:\M{E}\rightarrow \M{E}'$ such that $f(\I)=\I'$ and for all $e\in \I$, $k(e)=k'(f(e))$.
\end{proposition}
\begin{proof}
		We say that $\M{E}=\LRBES{}$ and $\M{E}'=(E',F',\mapsto',\cf',\rhd',\lambda',\mathsf{Act}')$ and do a case analysis on the Structural congruence rules:
		\begin{description}
			\item[$P=P_0\vert P_1$ and $P'=P_1 \vert P_0$:] Products are unique up to isomorphism, and 			
			$$f(e)=\begin{dcases*}
			(e_1,e_0) & if $e=(e_0,e_1)$ \\
			(e_1,*) & if $e=(*,e_1)$ \\
			(*,e_0) & if $e=(e_0,*)$ \\
			\end{dcases*}$$ clearly fulfils the conditions other conditions and remains a morphism after the enablings and preventions describing the link dependencies are added to the product.
			\item[$P= P_0\vert (P_1 \vert P_2)$ and $P'= (P_0\vert P_1 ) \vert P_2$:] Products are associative up to isomorphism, and $f((e_0,(e_1,e_2))=
			((e_0,e_1),e_2)$ clearly fulfils the other conditions and remains a morphism after the enablings and preventions describing the link dependencies are added to the product. 
			\item[$P=P'\vert 0$:] If $f((e,*))=e$, then this clearly holds.
			\item[$P=P_0 + P_1$ and $P'= P_1 + P_0$:] Coproducts are unique up to isomorphism, and $f(i,e)=(1-i,e)$ clearly fulfils the other conditions.
			\item[$P= P_0 + (P_1 + P_2) $ and $P'= (P_0 + P_1) + P_2$:] Coproducts are associative up to isomorphism, and $f((e_0,(e_1,e_2)))=((e_0,e_1),e_2)$ clearly fulfils the other conditions. 
			\item[$P=P'+ 0$:] Clearly $f(0,e)=e$ is an isomorphism, $\I=\{0\}\times\I'$, and $k(0,e)=k'(e)$.
			\item[${P}={!Q}$ and ${P'}={!Q\vert Q}$:] Obvious.
		\end{description}
	\end{proof}


\begin{proof}[Proof of Theorem~\ref{the:PtoLRBEStrans}]
Let $\M{E}=\LRBES{}$ and $\M{E}'=(E',F',\mapsto',\cf',\rhd',\lambda',\mathsf{Act}')$.
We prove the theorem by induction on $P\xrightarrow{\mu[m]} P'$:
\begin{enumerate}
\item Suppose $P=\inp{a}{x}.Q$, $P'=\inp{a}{x}[m].Q[x:= b_{[n]}]$, $\mathsf{std}(Q)$, and $\mu=\inp{a}{b}$. Then for all $n\in(\M{N}\setminus \mathsf{sbn}(Q))=(\M{N}\setminus \mathsf{sbn}(Q[x:= b_{[n]}]))$, we have $\lrBraces{Q[x:=n]}=\lrangles{\M{E}_n,\I_n,k_n}$, $\lrBraces{Q[x:=b_{[n]}][b_{[n]}:=n]}=\lrangles{\M{E}_n',\I_n',k_n'}$, and an isomorphism $f_n:\M{E}_n\rightarrow\M{E}_n'$. We define our isomorphism \[f((n,e_n))=\begin{dcases*}(n,f_n(e_n)) & if $e_n\in E_n$ \\
(n,e_n') & for $\{e_n'\}=\{e'\mid (n,e')\in E' \text{ and } e'\notin E_n'\}$ otherwise \\
\end{dcases*}\] Since all bound names are different from all other bound and free names, $b\notin \bn(Q)$, and therefore there exists an $e\in E$ such that $\lambda(e)=a(b)$, and for all $e'\in E$ either $e'=e$, $e'\cf e$, or $\{e\}\mapsto e'$. We therefore get $\I=\emptyset\xrightarrow{\{e\}} X$ and $f(X)=\I'$, and the rest of the conditions fulfilled.
\item Suppose $P=\outp{a}{x}.Q$, $P'=\outp{a}{x}[m].Q$, $\mu=\outp{a}{x}$, and $\mathsf{std}(Q)$. This case is similar to the previous, without the choice of substitutions.
\item Suppose $P=\alpha[n].Q$, $P'=\alpha[n].Q'$, $Q\xrightarrow{\mu[m]} Q'$, $m\neq n$, and if $\mu=\outp{a}{x}$ then $x\notin n(\alpha)$. Then let $\lrBraces{Q}=\lrangles{\M{E}_Q,\I_Q,k_Q}$ and $\lrBraces{Q'}=\lrangles{\M{E}_Q',\I_Q',k_Q'}$. We have an isomorphism $f_Q:\M{E}_Q\rightarrow \M{E}_Q'$ and a transition $\I_Q\xrightarrow{e_Q} X_Q$ such that $\lambda_Q(e_Q)=\mu$, $f\circ k_Q'=k_Q[e_Q\mapsto m]$, and $f(X_Q)=\I_Q'$. We define our isomorphism \[f((n,e_n))=\begin{dcases*}(n,f_n(e_n)) & if $e_n\in E_n$ \\
(n,e_n') & for $\{e_n'\}=\{e'\mid (n,e')\in E' \text{ and } e'\notin E_n'\}$ otherwise \\
\end{dcases*}\] and $e=(x,e_Q)$ if $\alpha=\inp{a}{x}$, and \[f(e')=\begin{dcases*}f_Q(e') & if $e'\in E_Q$ \\
e'' & for $\{e''\}=\{e'''\mid e'''\in E' \text{ and } e'''\notin E_Q'\}$ otherwise \\
\end{dcases*}\] and $e=e_Q$ if $\alpha=\outp{a}{x}$. These clearly fulfil the conditions.
\item Suppose $P=P_0\vert P_1$, $P'=P_0'\vert P_1$, $P_0\xrightarrow{\mu[m]} P_0'$, $\mathsf{fsh}[n](P_1)$, and if $\mu=\outp{a}{x}$ then $x\notin \fn(P_1)$. Then let $\lrBraces{P_0}=\lrangles{\M{E}_0,\I_0,k_0}$, $\lrBraces{P_0'}=\lrangles{\M{E}_0',\I_0',k_0'}$, and $\lrBraces{P_1}=\lrangles{\M{E}_1,\I_1,k_1}$. We then have an isomorphism $f_0:\M{E}_0\rightarrow \M{E}_0'$ and transition $\I_0\xrightarrow{e_0} X_0$ such that $\lambda_0(e_0)=\mu$, $f_0\circ k_0'=k_0[e_0\mapsto m]$, and $f_0(X_0)=\I_0'$. We define our isomorphism \[f(e')=\begin{dcases*}
(f_0(e_0'),*) & if $e'=(e_0',*)$ \\
(*,e_1') & if $e'=(*,e_1')$ \\
(f_0(e_0'),e_1') & if $e'=(e_0',e_1')$ \\
\end{dcases*} \] and $e=(e_0,*)$.
Since $\mathsf{sbn}(P_0)=\mathsf{sbn}(P_0')$ this is an isomorphism, and since all free and bound names are different, $\mathsf{no}(\mu)\cap \mathsf{sbn}(P_1)=\emptyset$, implying $\I\xrightarrow{e}$. The other conditions are clearly fulfilled.
\item Suppose $P=P_0\vert P_1$, $P'=(\nu x) (P_0'\vert P_1)$, $\mu=\tau$, $P_0\xrightarrow{\inp{a}{x}[m]} P_0'$, and $P_1\xrightarrow{\outp{a}{x}[m]} P_1'$. Then let  $\lrBraces{P_0}=\lrangles{\M{E}_0,\I_0,k_0}$, $\lrBraces{P_0'}=\lrangles{\M{E}_0',\I_0',k_0'}$, $\lrBraces{P_1}=\lrangles{\M{E}_1,\I_1,k_1}$, and $\lrBraces{P_1'}=\lrangles{\M{E}_1',\I_1',k_1'}$. Then we have isomorphisms $f_0:\M{E}_0\rightarrow \M{E}_0'$ and $f_1:\M{E}_1\rightarrow \M{E}_1'$ and transitions $\I_0\xrightarrow{e_0} X_0$ and $\I_1\xrightarrow{e_1} X_1$ such that $\lambda_0(e_0)=\inp{a}{x}$, $\lambda_1(e_1)=\outp{a}{x}$, $f_0\circ k_0'=k_0[e_0\mapsto m]$, $f_1\circ k_1'=k_1[e_1\mapsto m]$, $f_0(X_0)=\I_0'$, and $f_1(X_1)=\I_1'$. We then define our isomorphism \[f(e')=\begin{dcases*}
(f_0(e_0'),*) & if $e'=(e_0',*)$ \\
(*,f_1(e_1')) & if $e'=(*,e_1')$ \\
(f_0(e_0'),f_1(e_1')) & if $e'=(e_0',e_1')$ \\
\end{dcases*} \] and $e=(e_0,e_1)$. Since $\mathsf{sbn}(P_0)=\mathsf{sbn}(P_0')$ and the existence of $(f_0(e_0),f_1(e_1))\in \I'$ and $\inp{a}{x}[m]$ and $\outp{a}{x}[m]$ prevents $(\nu x)$ from affecting $\M{E}'$, $f$ is an isomorphism, and since $\mathsf{no}(\tau)=\emptyset$, we have a transition $\I\xrightarrow{e}$. The other conditions are clearly fulfilled.
\item Suppose $P=P_0+P_1$, $P'=P_0'+P_1$, $P_0\xrightarrow{\mu[m]} P_0'$, and $\mathsf{std}(P_1)$. Then let $\lrBraces{P_0}=\lrangles{\M{E}_0,\I_0,k_0}$, $\lrBraces{P_0'}=\lrangles{\M{E}_0',\I_0',k_0'}$, and $\lrBraces{P_1}=\lrangles{\M{E}_1,\I_1,k_1}$. We then have an isomorphism $f_0:\M{E}_0\rightarrow \M{E}_0'$ and transition $\I_0\xrightarrow{e_0} X_0$ such that $\lambda_0(e_0)=\mu$, $f_0\circ k_0'=k_0[e_0\mapsto m]$, and $f_0(X_0)=\I_0'$. We define out isomorphism \[f((i,e_i))= \begin{dcases*}
(0,f_0(e_0)) & if i=0 \\
(1,e_1) & if i=1 \\
\end{dcases*}\] and $e=(0,e_0)$. Isomorphism is preserved by the coproduct, and the remaining conditions are clearly fulfilled.
\item Suppose $P=(\nu x) Q$, $P'=(\nu x) Q'$, $Q\xrightarrow{\mu[m]} Q'$, and $x\notin n(\mu)$. Then let $\lrBraces{Q}=\lrangles{\M{E}_Q,\I_Q,k_Q}$ and $\lrBraces{Q'}=\lrangles{\M{E}_Q',\I_Q',k_Q'}$. We have an isomorphism $f_Q:\M{E}_Q\rightarrow \M{E}_Q'$ and a transition $\I_Q\xrightarrow{e_Q} X_Q$ such that $\lambda_Q(e_Q)=\mu$, $f\circ k_Q'=k_Q[e_Q\mapsto m]$, and $f(X_Q)=\I_Q'$. Either there exist past actions $\inp{b}{a}[m]$ and $\outp{b}{a}[m]$ in $P$ which are not guarded by a restriction $(\nu a)$ in $P$ or not. If such $\inp{b}{a}[m]$ and $\outp{b}{a}[m]$ exist, then $\lrangles{\M{E},\I,k}=\lrangles{\M{E}_Q,\I_Q,k_Q}$ and $\lrangles{\M{E}',\I',k'}=\lrangles{\M{E}_Q',\I_Q',k_Q'}$, and the rest follows trivially. Otherwise restriction preserves morphisms, and clearly does not affect $e=e_Q$.
\end{enumerate}
\end{proof}

\begin{proof}[Proof of Theorem~\ref{the:PtoLRBEStrans2}]
We prove this by structural induction on $P$:
\begin{itemize}
\item Suppose $P=0$. Then $E=\emptyset$, and no transition $\I\xrightarrow{\{e\}} X$ exists.
\item Suppose $P=\outp{a}{x}.Q$. Let $\lrBraces{Q}_{\M{N}}=\lrangles{\M{E}_Q,\I_Q,k_Q}$, $\M{E}_Q=\LRBES{Q}$, and $\M{E}'=(E',F',\mapsto',\cf',\rhd',\lambda',\mathsf{Act}')$. Then there exists $e$ such that $E\setminus E_Q=\{e\}$, and for all $e'\in E$, if $e'\neq e$ then $\{e\}\mapsto e'$. Therefore this is the only possible $e$ such that $\I\xrightarrow{\{e\}}$. Additionally we have $\lambda(e)=\outp{a}{x}$ and $P\xrightarrow{\outp{a}{x}[m]} \outp{a}{x}[m].Q$ for any key $m$, and the rest of the case is straightforward.
\item Suppose $P=\inp{a}{x}.Q$. 
Then there must exist $b\in \M{N}\setminus \mathsf{sbn}(P)$ such that $\lambda(e)=\inp{a}{b}$, and for all $e'\in E$ either $e=e'$, $e\cf e'$, or $\{e\}\mapsto e'$. There then exists a transition $P\xrightarrow{\inp{a}{b}[m]} \inp{a}{b}[m].Q[x:= b_{[m]}]$ and the rest of the case is straightforward.
\item Suppose $P=\outp{a}{x}[n].Q$. Let $\lrBraces{Q}_{\M{N}}=\lrangles{\M{E}_Q,\I_Q,k_Q}$, $\M{E}_Q=\LRBES{Q}$, and $\M{E}'=(E',F',\mapsto',\cf',\rhd',\lambda',\mathsf{Act}')$. Then $\I\xrightarrow{\{e\}} X$ implies $\I_Q\xrightarrow{\{e\}} X\cap E_Q$. We therefore have a transition $Q\xrightarrow{\mu[m]} Q'$ such that $\lrBraces{Q'}_{\M{N}}=\lrangles{\M{E}_Q',\I_Q',k_Q'}$ and an isomorphism $f_Q:\M{E}_Q\rightarrow \M{E}_Q'$ such that $\lambda_Q(e)=\mu$, $f_Q\circ k_Q'=k_Q[e\mapsto m]$, and $f_Q(X\cap E_Q)=\I_Q'$. This gives us a transition $P\xrightarrow{\mu[m]} \outp{a}{x}[n].Q'$ and the rest of the case is straightforward.
\item Suppose $P=\inp{a}{x}[n].Q$. This case is a combination of the previous two.
\item Suppose $P=Q+R$ and let $\lrBraces{Q}_{\M{N}}=\lrangles{\M{E}_Q,\I_Q,k_Q}$, $\M{E}_Q=\LRBES{Q}$, $\lrBraces{R}_{\M{N}}=\lrangles{\M{E}_R,\I_R,k_R}$, $\M{E}_R=\LRBES{R}$, and $\M{E}'=(E',F',\mapsto',\cf',\rhd',\lambda',\mathsf{Act}')$. Either $e=(0,e_Q)$ or $e=(1,e_R)$. In the first case we get a transition transition $Q\xrightarrow{\mu[m]} Q'$ such that $\lrBraces{Q'}_{\M{N}}=\lrangles{\M{E}_Q',\I_Q',k_Q'}$ and an isomorphism $f_Q:\M{E}_Q\rightarrow \M{E}_Q'$ such that $\lambda_Q(e_Q)=\mu$, $f_Q\circ k_Q'=k_Q[e_Q\mapsto m]$, and $f_Q(\{e_Q'\vert (0,e_Q')\in X\})=\I_Q'$. We therefore define $$f(e)=\begin{dcases*}
(0,f_Q(e_Q)) & if $e=(0,e_Q)$ \\
e & otherwise \\
\end{dcases*}$$ and the rest of the case is straightforward. If $e=(1,e_R)$, the proof is similar.
\item Suppose $P=Q\vert R$ and let $\lrBraces{Q}_{\M{N}}=\lrangles{\M{E}_Q,\I_Q,k_Q}$, $\M{E}_Q=\LRBES{Q}$, $\lrBraces{R}_{\M{N}}=\lrangles{\M{E}_R,\I_R,k_R}$, $\M{E}_R=\LRBES{R}$, and $\M{E}'=(E',F',\mapsto',\cf',\rhd',\lambda',\mathsf{Act}')$. Either $e=(e_Q,*)$, $e=(*,e_R)$, or $e=(e_Q,e_R)$. If $e=(e_Q,*)$ then we have a transition $Q\xrightarrow{\mu[m]} Q'$ such that $\lrBraces{Q'}_{\M{N}}=\lrangles{\M{E}_Q',\I_Q',k_Q'}$ and an isomorphism $f_Q:\M{E}_Q\rightarrow \M{E}_Q'$ such that $\lambda_Q(e_Q)=\mu$, $f_Q\circ k_Q'=k_Q[e_Q\mapsto m]$, and $f_Q(\{e_Q'\vert (e_Q',*)\in X \text{ or } (e_Q',e_R')\in X\})=\I_Q'$. We therefore get $P\xrightarrow{\mu[m]} Q'\vert R$ so long as $\mathsf{fsh}[m](R)$, and if not we can do the same with a different $m$. We can then define \[f(e')=\begin{dcases*}
(f_0(e_0'),*) & if $e'=(e_0',*)$ \\
(*,e_1') & if $e'=(*,e_1')$ \\
(f_0(e_0'),e_1') & if $e'=(e_0',e_1')$ \\
\end{dcases*} \] and the rest of the case is straightforward. If $e=(*,e_R)$, the case is similar. If $e=(e_Q,e_R)$, then we have transition $Q\xrightarrow{\alpha[m]} Q'$ such that $\lrBraces{Q'}_{\M{N}}=\lrangles{\M{E}_Q',\I_Q',k_Q'}$ and isomorphism $f_Q:\M{E}_Q\rightarrow \M{E}_Q'$ such that $\lambda_Q(e_Q)=\alpha$, $f_Q\circ k_Q'=k_Q[e_Q\mapsto m]$, and $f_Q(\{e_Q'\vert (e_Q',*)\in X \text{ or } (e_Q',e_R')\in X\})=\I_Q'$, and transition $R\xrightarrow{\alpha'[m]} R'$ such that $\lrBraces{R'}_{\M{N}}=\lrangles{\M{E}_R',\I_R',k_R'}$ and isomorphism $f_R:\M{E}_R\rightarrow \M{E}_R'$ such that $\lambda_R(e_R)=\alpha'$, $f_R\circ k_R'=k_R[e_R\mapsto m]$, and $f_R(\{e_R'\vert (*,e_R')\in X \text{ or } (e_Q',e_R')\in X\})=\I_R'$ and there exist names $a,b$ such that either $\alpha=\inp{a}{b}$ and $\alpha'=\outp{a}{b}$ or $\alpha'=\inp{a}{b}$ and $\alpha=\outp{a}{b}$. We therefore get a transition $P\xrightarrow{\tau[m]} (\nu b) (Q'\vert R')$ and define \[f(e)=\begin{dcases*}
(f_0(e_0'),*) & if $e=(e_0',*)$ \\
(*,f_1(e_1')) & if $e=(*,e_1')$ \\
(f_0(e_0'),f_1(e_1')) & if $e=(e_0',e_1')$ \\
\end{dcases*} \] and the rest of the case is straightforward.
\item Suppose $P=(\nu a) Q$. Let $\lrBraces{Q}_{\M{N}}=\lrangles{\M{E}_Q,\I_Q,k_Q}$, $\M{E}_Q=\LRBES{Q}$, and $\M{E}'=(E',F',\mapsto',\cf',\rhd',\lambda',\mathsf{Act}')$. Then either there exist past actions $\inp{b}{a}[m]$ and $\outp{b}{a}[m]$ in $Q$ which are not guarded by a restriction $(\nu a)$ in $Q$ or not. If such $\inp{b}{a}[m]$ and $\outp{b}{a}[m]$ do exist, then they must be in parallel, and therefore there exists an event $e'\in E\setminus \I$ such that $\lambda(e)=\outp{b}{a}$, and for all other events $e''\in E$, if $\lambda(e')$ outputs $a$ then $e'=e$, and if $a\in \mathsf{no}(\lambda(e'))$ then $\{e'\}\mapsto e''$. Additionally there exists $e'''\in \I$ such that $e'''\cf e'$ and $\lambda(e''')=\tau$. We therefore get that $a\notin n(e)$. Additionally $\I_Q=\I\xrightarrow{\{e\}} X$ and by induction we have a transition $Q\xrightarrow{\mu[m]} Q'$ such that $\lrBraces{Q'}_{\M{N}}=\lrangles{\M{E}_Q',\I_Q',k_Q'}$ and an isomorphism $f_Q:\M{E}_Q\rightarrow \M{E}_Q'$ such that $\lambda_Q(e)=\mu$, $f_Q\circ k_Q'=k_Q[e\mapsto m]$, and $f_Q(X)=\I_Q'=\I'$. We define $f=f_Q$ and the result follows.
If no such $\inp{b}{a}[m]$ and $\outp{b}{a}[m]$ exist in $Q$ then clearly $a\notin n(\lambda(e))$, and restriction preserves morphisms, meaning the proof is straightforward.
\end{itemize}
\end{proof}

\end{document}